\documentclass[amsmath,amssymb,aps,prx,twocolumn,notitlepage,superscriptaddress,longbibliography]{revtex4-2}
\usepackage{graphicx}
\usepackage{dcolumn}
\usepackage{bm}
\usepackage{braket}
\usepackage{amsthm}
\usepackage{epstopdf}
\usepackage{amssymb}
\usepackage{mathtools}
\usepackage{amsmath}
\usepackage{hyperref}
\usepackage{comment}
\usepackage{xcolor}
\usepackage[ruled,vlined]{algorithm2e}
\usepackage{array}
\usepackage{tabularx}
\usepackage[normalem]{ulem}
\usepackage{multirow}
\usepackage{dsfont}
\newcolumntype{P}[1]{>{\centering\arraybackslash}p{#1}}

\newcommand{\Haf}{\text{Haf}}
\newcommand{\Per}{\text{Per}}

\newcommand{\Tr}{\text{Tr}}

\theoremstyle{definition}

\newtheorem{theorem}{Theorem}

\newtheorem{lemma}{Lemma}

\begin{document}
\title{Approximating outcome probabilities of linear optical circuits}
\author{Youngrong Lim}
\email{sshaep@kias.re.kr}
\affiliation{School of Computational Sciences, Korea Institute for Advanced Study, Seoul 02455, Korea}
\author{Changhun Oh}
\email{changhun@uchicago.edu}
\affiliation{Pritzker School of Molecular Engineering, University of Chicago, Chicago, Illinois 60637, USA}
\begin{abstract}
   Quasiprobability representations are important tools for analyzing a quantum system, such as a quantum state or a quantum circuit.
    In this work, we propose classical algorithms specialized for approximating outcome probabilities of a linear optical circuit using quasiprobability distributions. Notably, we can reduce the negativity bound of a circuit from exponential to at most polynomial for specific cases by modulating the shapes of quasiprobability distributions thanks to the symmetry of the linear optical transformation in the phase space. 
    Consequently, our scheme provides an efficient estimation of outcome probabilities within an additive-error whose precision depends on the classicality of the input state. 
    When the classicality is high enough, we reach a polynomial-time estimation algorithm of a probability within a multiplicative-error by an efficient sampling from a log-concave function.
    By choosing appropriate input states and measurements, our results provide plenty of quantum-inspired classical algorithms for approximating various matrix functions beating best-known results. Moreover, we give sufficient conditions for the classical simulability of Gaussian boson sampling using our approximating algorithm for any (marginal) outcome probability under the poly-sparse condition. 
\end{abstract}
\maketitle

\section{introduction}
Quantum computers are believed to provide significant advantages in solving computational problems beyond classical power~\cite{shor1999polynomial,arute2019quantum}.
Despite the potential advantage, determining which types of quantum circuits are classically simulable is still a longstanding open problem.
A viable approach for the problem is to examine quasiprobability distributions in the phase space, a standard method in quantum optics~\cite{cahill1969density}.
Specifically, on the one hand, if the output distributions of a quantum circuit can be described by nonnegative quasiprobability distributions, there is a classically efficient simulation of the circuit~\cite{mari2012positive}. 
On the other hand, one can estimate the outcome probabilities of a quantum circuit with the convergence rate depending on the negativity bound of the circuit, a measure of negativity in the phase space~\cite{pashayan2015estimating}. 
In both cases, the negativity of quasiprobability distributions plays a crucial role, which has been extensively studied~\cite{veitch2012negative,takagi2018convex,garcia2020efficient,tan2020negativity}.

Meanwhile, boson sampling has recently attracted lots of attention due to its feasible quantum advantage using a linear optical circuit~\cite{aaronson2011computational,zhong2020quantum,zhong2021phase,madsen2022quantum,deng2023gaussian}.
Also, there have been numerous studies of boson sampling using quasiprobability distributions~\cite{rahimi2015can,rahimi2016sufficient,opanchuk2018simulating,qi2020regimes,drummond2022simulating}.
In particular, while Fock-state and Gaussian boson sampling are believed to be hard to classically simulate~\cite{aaronson2011computational,hamilton2017gaussian}, boson sampling with a classical state input, represented by a nonnegative $P$-function, is efficiently simulated because its output distribution can be expressed by nonnegative quasiprobability distributions~\cite{rahimi2016sufficient}.
Interestingly, this quantum optical observation leads to the fact that the permanent of Hermitian positive-semidefinite (HPSD) matrices can be approximated within a multiplicative-error more easily (in $\text{BPP}^\text{NP}$) than that of arbitrary matrices (\#P-hard)~\cite{valiant1979complexity,rahimi2015can}.
In addition, a quantum-inspired algorithm using quasiprobability distributions has been proposed for estimating the permanent of an HPSD matrix within an additive-error, which outperforms Gurvits' algorithm under certain eigenvalue conditions~\cite{gurvits2005complexity,chakhmakhchyan2017quantum}. 
Thus, studying the quasiprobability representation of quantum circuits often provides new insight into computational problems.

However, previous studies are limited to the cases of classical input states with no negativity~\cite{rahimi2015can,rahimi2016sufficient,chakhmakhchyan2017quantum,qi2020regimes} or negativity bound that is polynomial in the system size~\cite{pashayan2015estimating}, which cannot cover typical quantum circuits with exponential negativity bound. A central question is how to generalize the quasiprobability methods to handle cases of exponential negativity bound. Such a generalization can also lead to improved quantum-inspired algorithms for approximating matrix functions, e.g., permanent and hafnian. Moreover, it is still open to finding an efficient algorithm 
for the multiplicative-error approximation of a matrix function beyond the method using sampling from nonnegative quasiprobability distributions~\cite{rahimi2015can}. Since a multiplicative-error approximation is significantly more powerful than an additive-error one and only a few examples have been known~\cite{jerrum2004polynomial,barvinok2016computing,barvinok2021remark}, such findings have intriguing applications in computational complexity.

\begin{table*}[!ht]
\setlength{\tabcolsep}{0.5em}
{\renewcommand{\arraystretch}{2}
\begin{ruledtabular}
\begin{tabularx}{\textwidth}{P{1.5cm}||P{7.5cm}|P{8.5cm}}
     & additive-error &  multiplicative-error (Eq.~(\ref{eq:muldef}))   \\
    \hline
    \hline
 $|\text{Haf}(R)|^2$ & $\epsilon \prod^M_{i} \frac{\lambda_{\max}^2}{\sqrt{\lambda^{2}_{\max}(W(1/e)-1)^2-\lambda^{2}_i W(1/e)^2}}~(*)~(\text{Th}.~\ref{th:haf})$  & \#P-hard~\cite{barvinok2016combinatorics} \\
 \hline
 \multirow{2}{3em}{Per($B$)} & $\lambda_{\min}=0: \epsilon \prod_i^M \frac{4\lambda_{\max}^2}{e(2\lambda_{\max}-\lambda_i)}~(*)$~(Th.~\ref{th:per}) & $\lambda_{\min}=0 : \text{NP-hard}$~\cite{meiburg2021inapproximability}  \\
 &$\lambda_{\min}>0:\epsilon \prod^M_i H^B_i(\lambda_i)~(*)~\text{Eq.~(\ref{eq:per2})}$ &$\lambda_{\min}>0:\frac{\lambda_{\max}}{\lambda_{\min}}\leq 2$~\cite{barvinok2021remark}~$(*)$   \\
 \hline
 Haf($A$) &$\epsilon \prod_{i}^M H_i^A(n,r_i)~(*)$~Eq.~(\ref{eq:addhafa}) &$n \geq \frac{\left(6 \sinh (2 r_{\max})+\sqrt{18 \cosh (4 r_{\max})-14}-2\right)}{4}~(*)$~(Th.~\ref{th:FPRAS})\\
 \hline
 Tor($R'$) & $\epsilon \prod_{i}^M T_i(\lambda_i)~(*)$~Eq.~(\ref{eq:addtorr}) & ?  \\
 \hline
 Tor($B'$) & $\epsilon \prod_{i}^M T_i^B(\lambda_i)~(*)$~Eq.~(\ref{eq:addtorb}) &  $\lambda_{\min} \geq \frac{1}{2}$ and $\lambda_{\max} \leq \frac{-\lambda_{\min}^2+3\lambda_{\min}-1}{\lambda_{\min}}~(*)$~Eq.~(\ref{eq:torcmulti})   \\
 \hline
 Tor($A'$) &$\epsilon \prod_{i}^M T_i^A(n,r_i)~(*)$~Eq.~(\ref{eq:addtora}) & $n\geq \frac{1}{2}\left(e^{2r_{\max}}\sqrt{e^{8r_{\max}}+3}+e^{6r_{\max}}-1\right)~(*)$~Eq.~(\ref{eq:tormul}) \\ 
\end{tabularx}
\end{ruledtabular}
}
\caption{Precision and conditions of efficient algorithms for estimating various matrix functions. $R $: complex symmetric matrices, $B$: HPSD matrices, $R'=\begin{pmatrix} 0 & R^* \\ R & 0\end{pmatrix}$,  $B'=\begin{pmatrix} B^T & 0 \\ 0 & B\end{pmatrix}$, $A=\begin{pmatrix} R & B \\ B^T & R^*\end{pmatrix}$, $A'=\begin{pmatrix} B^T & R^* \\R & B\end{pmatrix}$. $(*)$ indicates the results in the present work.
A question mark represents unknown.}
\label{table:estimation}
\end{table*}

In this work, we provide algorithms specialized for approximating the outcome probabilities of a linear optical circuit. First, we choose the $s$-parameterized quasiprobability distributions ($s$-PQDs), a  generalization of $P$-, $Q$-, and Wigner distributions~\cite{glauber1963quantum}. An advantage of adopting $s$-PQDs over previous approaches is that we can always obtain a nonnegative representation for a Gaussian input state by choosing an appropriate parameter $s$. Consequently, for a Gaussian boson sampling (GBS) circuit, we can significantly reduce the negativity bound of the circuit, which depends only on the maximum peak of the $s$-PQDs of measurement operators. Furthermore, the negativity bound is determined by the maximum possible $s$, called classicality of the input Gaussian state~\cite{lee1991measure}, which has a clear physical meaning: the more classicality in the input state, the lower negativity bound. 

Our main technical contribution is to provide a way of manipulating the shape of $s$-PQDs using the symmetry of the circuit transformation in the phase space. This method can considerably lower the negativity bound, from exponential to at most polynomial in several cases, which renders efficient estimations of outcome probabilities within 1/poly additive-error. Strikingly, when the classicality of the input state is high enough, we introduce a fully polynomial-time randomized approximation scheme (FPRAS) by an efficient sampling from a log-concave function~\cite{lovasz2007geometry,barvinok2021remark}, which can efficiently (in BPP) approximate the corresponding outcome probability within a multiplicative-error. 

Our results have several intriguing applications to problems in computational complexity. First, we give estimating algorithms with additive-errors for matrix functions represented by the outcome probabilities of a linear optical circuit, beating the best-known classical algorithms, e.g., the hafnian of a complex symmetric matrix and the permanent of an HPSD matrix. Second, we provide efficient multiplicative-error algorithms for those matrix functions with certain structured matrices, which are not considered in the literature to the best of our knowledge.  Last but not least, we present sufficient conditions on the classical simulability of GBS, by applying our estimation scheme to any (marginal) outcome probability of a GBS circuit with 1/poly additive-error under a poly-sparsity condition.

\section{Summary of results}
We place a background material in Section~\ref{sec:back} to provide preliminary information of quasiprobability and estimation schemes. The main results are following: 

\begin{itemize}
    \item Our key technique for manipulating quasiprobability distribution in the phase space (Section~\ref{sec:res.mani})
    \item Scheme for additive/multiplicative-errors approximations of outcome probabilities of a linear optical circuit depending on the parameter regime (Section~\ref{sec:res.improved})
    \item Providing additive-error approximation algorithms for various matrix functions beating best-known classical algorithms (Section~\ref{sec:add}), such as the hafnian of a complex symmetric matrix~(Theorem~\ref{th:haf}) and the permanent of HPSD matrix~(Theorem~\ref{th:per})
    \item Providing multiplicative-error approximation algorithms for various matrix functions  (Section~\ref{sec:mul}), such as the hafnian of a structured matrix~(Theorem~\ref{th:FPRAS})
    \item Sufficient conditions for the efficient simulability of lossy Gaussian boson sampling under a poly-sparsity assumption (Section~\ref{sec:sparse} and Theorem~\ref{th:GBS})
\end{itemize}
  We summarize additive- and multiplicative-errors quantum-inspired algorithms for various matrix functions including Torontonian in Table~\ref{table:estimation}.

\section{Background}\label{sec:back}
\subsection{Quasiprobability representation of outcome probability}

Let us consider the Born rule probabilities of a quantum optical circuit using quasiprobability distributions in the phase space. 
For an $M$-mode input state $\rho_{\text{in}}$, a quantum channel ${\cal E}$, and a measurement $\Pi_{\bm{\nu}}$, the probability for a measurement outcome $\bm{\nu}\coloneqq (\nu_1,\dots,\nu_M)$ can be written as~\cite{rahimi2016sufficient}
\begin{align}\label{eq:probgen}
    p(\bm{\nu})&=\pi^M\int d^{2M}\bm{\alpha}d^{2M}\bm{\beta} W^{(t)}_{\rho_\text{in}}(\bm{\alpha})T^{(t,s)}_{\cal E}(\bm{\beta}|\bm{\alpha})W^{(-s)}_{\Pi_{\bm{\nu}}}(\bm{\beta}),
\end{align}
where $W^{(t)}_{\rho_{\text{in}}}(\bm{\alpha})$, $W^{(-s)}_{\Pi_{\bm{\nu}}}(\bm{\beta})$, and $T^{(t,s)}_{\cal E}(\bm{\beta}|\bm{\alpha})$ are $s$-PQDs of the input state, measurement, and the transition function of the circuit channel, respectively. Here, $\bm{\alpha}$ is a quadrature variable in the $2M$-dimensionnal phase space and $\bm{\beta}$ is a transformed quadrature variable by the transition function $T^{(t,s)}_{\cal E}(\bm{\beta}|\bm{\alpha})$. Specifically, $W^{(s)}_\rho(\bm{\alpha})$ is the $s$-PQD for a Hermitian operator $\rho$ defined by
\begin{equation}
    W^{(s)}_\rho(\bm{\alpha})=\int \frac{d^{2M}\bm{\alpha}'}{\pi^{2M}}\Tr[\rho D(\bm{\alpha}')]e^{\bm{\alpha}'s\bm{ \alpha}'^{\dagger}/2}e^{\bm{\alpha \alpha}'^{\dagger}-\bm{\alpha'\alpha}^{\dagger}},
\end{equation}
where $D(\bm{\alpha}')=e^{\bm{\alpha' \hat{a}}^{\dagger}-\bm{\hat{a} \alpha'}^{\dagger}}$ is the $M$-mode displacement operator. Note that $s=-1,0,1$ of $s$-PQDs correspond to the Q-, Wigner, and P-distribution, respectively. Also, $(-s)$-PQDs of the measurement operators satisfy the normalization condition such that $\pi^M \sum_{\bm{\nu}}W^{(-s)}_{\Pi_{\bm{\nu}}}(\bm{\beta})=1$ for any $\bm{\beta}$. The transition function of a quantum channel ${\cal E}$ is defined by~\cite{rahimi2016sufficient}
\begin{align}
&T^{(t,s)}_{\cal E}(\bm{\beta}|\bm{\alpha})=\int \frac{d^{2M}\bm{\zeta}}{\pi^{2M}}e^{\bm{\zeta} \bm{s}\bm{\zeta}^{\dagger}/2}e^{\bm{\beta}\bm{\zeta}^{\dagger}-\bm{\zeta}\bm{\beta}^{\dagger}} \nonumber \\
&\times \int \frac{d^{2M}\bm{\xi}}{\pi^{M}}e^{-\bm{\xi} \bm{t}\bm{\xi}^{\dagger}/2}e^{\bm{\xi}\bm{\alpha}^{\dagger}-\bm{\alpha}\bm{\xi}^{\dagger}}\Tr \left[ {\cal E}(D^{\dagger}(\bm{\xi}))D(\bm{\zeta}) \right],
\end{align}
where 
\begin{equation}
    {\cal E}(D^{\dagger}(\bm{\xi}))=e^{\bm{\xi}\bm{\xi}^{\dagger}/2}\int \frac{d^{2M}\bm{\gamma}}{\pi^M}e^{\bm{\gamma}\bm{\xi}^{\dagger}-\bm{\xi}\bm{\gamma}^{\dagger}}{\cal E}(\ket{\bm{\gamma}}\bra{\bm{\gamma}}),
\end{equation}
for a coherent state $\ket{\bm{\gamma}}$.
In this work, we are concerned with a linear optical circuit represented by a unitary matrix $U$ with a product input state $\rho_{\text{in}}=\otimes_{i=1}^M \rho_i$ and a product measurement $\Pi_{\bm{\nu}}=\otimes_{j=1}^M \ket{\nu_j}\bra{\nu_j}$. In that case, if we choose $t=s$, then the transition function becomes a simple form such as $T_{\cal E}(\bm{\beta}|\bm{\alpha})=\delta(\bm{\beta}-U\bm{\alpha})$. Consequently, Eq.~(\ref{eq:probgen}) can be written in the simpler form
\begin{align}\label{eq:probprod}
    p(\bm{\nu})&=\pi^M\int d^{2M}\bm{\alpha} \prod^M_{i=1}  W_{\rho_i}^{(s)}(\alpha_i)\prod^M_{j=1} W^{(-s)}_{\Pi_{\nu_j}}(\beta_j),
\end{align}
where $\beta_i=\sum_{j=1}^M U_{ji}\alpha_j$ by the transition function.
Suppose we have a product Gaussian input state with the covariance matrix $V_i$ and a product photon number measurement $\Pi_{m_j}$. The $s$-PQD of a single-mode Gaussian state with zero-displacement having covariance matrix $V$ is given by~\cite{adesso2014continuous}
\begin{equation}\label{eq:spqd}
   W_V^{(s)}(\bm{\alpha})=\frac{\exp \big[ -\bm{\alpha}(V-s\mathbb{I}_2/2)^{-1}\bm{\alpha}^T\big]}{\pi \sqrt{\det (V-s\mathbb{I}_2/2)}}.
\end{equation}
Note that for any physical covariance matrix $V$, there exists a critical value $s_{\max}$ such that $W_V^{(s)}(\bm{\alpha})$ is a proper Gaussian distribution for $s<s_{\max}$ and has a delta function singularity for $s=s_{\max}$~\cite{qi2020regimes}. 
We call $s_{\max}$ `classicality' of the Gaussian input state~~\cite{lee1991measure}.
Meanwhile, the $(-s)$-PQD for the photon number measurement operator $\ket{m}\bra{m}$ can be represented~\cite{wunsche1998some,tan2020negativity} as
\begin{equation}
W^{(-s)}_{\Pi_m}(\bm{\beta})=\frac{2}{\pi(s+1)}\left(\frac{s-1}{s+1} \right)^m \text{L}_m \left(\frac{4|\bm{\beta}|^2}{1-s^2} \right) e^{-\frac{2|\bm{\beta}|^2}{s+1}},
\end{equation}
where $\text{L}_m(x)$ is the $m^{\text{th}}$ Laguerre polynomial and $s>-1$.

\subsection{Approximation schemes of outcome probability}
\subsubsection{Additive-error approximation}
The first method is to estimate the Born rule probability within an additive-error. From the result in Ref.~\cite{pashayan2015estimating}, one can estimate the outcome probability $p(\bm{\nu})$ within error $\epsilon$ with success probability $1-\delta$ for a given number of samples $N=\frac{2{\cal M}_{\rightarrow}^2}{\epsilon^2}\log \frac{2}{\delta}$. Here, ${\cal M}_{\rightarrow}$ is the (forward) negativity bound of the circuit defined as
\begin{equation}
    {\cal M}_{\rightarrow}=\prod_{i=1}^M{\cal M}_{\rho_i}\prod^M_{j=1}\max_{\beta_j} \left|  W^{(-s)}_{\Pi_{\nu_j}}(\beta_j)\right|,
\end{equation}
where ${\cal M}_{\rho_i}=\int d^{2}\alpha_i \left| W_{\rho_i}^{(s)}(\alpha_i)\right|$ is the negativity of the state $\rho_i$ in the phase space. When the negativity bound ${\cal M}_{\rightarrow}$ increases at most polynomially in the number of modes $M$, we can efficiently estimate the corresponding outcome probability with additive-error $\epsilon$ within running time $T=\text{poly}(M,1/\epsilon,\log\delta^{-1})$.

Although this method is generally applicable for circuits having negativity, we usually encounter exponential negativity bound, i.e., ${\cal M}_\rightarrow=c^M$ with $c>1$, due to the input state negativity or a high maximum peak of the quasiprobability distribution of the measurement part. To circumvent this problem, finding a good quasiprobability representation for a given circuit with a small negativity is critical~\cite{zhu2016quasiprobability}. For instance, let us consider a GBS circuit with the input of a pure squeezed vacuum state and photon number measurement. By choosing $s$-PQDs, the negativity of the input state is removed when $s\leq s_{\max}$. However, the negativity bound is still exponential in the number of modes because of the high peaks of $s$-PQDs for the measurement operators except for a small squeezing parameter.

\subsubsection{Multiplicative-error approximation}
Remarkably, for special cases, we can reach a much stronger approximation, namely FPRAS.
An FPRAS can estimate the target function $p(\bm{\nu})$ within a multiplicative-error $\epsilon$, which means for any $0<\epsilon<1,~ 0<\delta<1$, the algorithm outputs $\mu$ such that
\begin{equation}\label{eq:muldef}
\text{Pr}[(1-\epsilon)p(\bm{\nu}) \leq \mu \leq (1+\epsilon)p(\bm{\nu})]\geq 1-\delta,
\end{equation}
with the running time  $\text{poly}(M,1/\epsilon,\log\delta^{-1})$. 
A sufficient condition for FPRAS is that the target function can be written as an integral of a log-concave function $f(t)$, i.e., $\log{f(\theta x+(1-\theta)y)\geq \theta \log{f(x)}+(1-\theta)\log{f(y)}}$ for all $x, y$ and $0<\theta<1$, such that $p(\bm{\nu})=\int_{\mathbb{R}^{2M}}f(t)dt$~\cite{lovasz2007geometry,barvinok2021remark}. Let us consider a linear optical circuit as in Eq.~(\ref{eq:probprod}). Since a multivariate Gaussian distribution is log-concave, $s$-PQD of a Gaussian input state fulfills the condition of log-concavity. If we choose a Gaussian measurement, the integrand satisfies the log-concavity but it is a trivial case. In this work, we develop a technique for making the quasiprobability of a non-Gaussian measurement log-concave, by manipulating quasiprobability in the phase space. 

\section{Results}
\subsection{Manipulating quasiprobability in the phase space}\label{sec:res.mani}
\begin{figure*}[t]
\includegraphics[width=\textwidth]{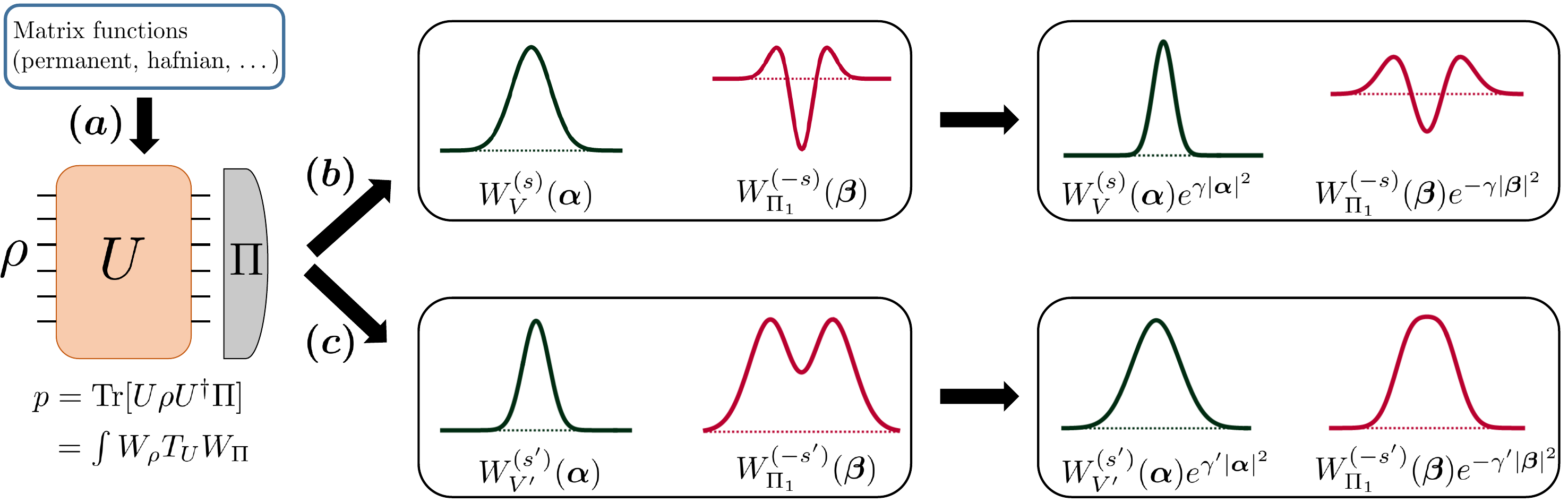}
\caption{Schematic diagram of quantum-inspired classical algorithms for approximating matrix functions. For a given matrix function, (a) find an embedding of the matrix function onto an outcome probability of a quantum circuit $(\rho, U, \Pi)$ and choose a quasiprobability representation of the probability. (b) We depict an example of a linear optical circuit, for the approximation scheme with additive-error. Using $s$-PQDs for the linear optical circuit, one can significantly reduce the negativity bound by appropriately choosing $\gamma<0$. (c) Approximation scheme with multiplicative-error. When the classicality of the input state is large, one can make the $s$-PQDs of the measurement operator a log-concave function by choosing a suitable $\gamma'>0$.
}
\label{fig:scheme}
\end{figure*}
In the previous section, we introduced two approximation schemes for the Born rule probability of an optical circuit. Those methods themselves do not seem to give powerful results for physically relevant situations. To obtain more interesting results, we propose a method manipulating the shapes of quasiprobability distributions using the symmetry of circuit transformation in the phase space. 

In a general situation, one can rewrite the Born rule probability Eq.~(\ref{eq:probgen}) as 
\begin{align}
p(\bm{\nu})&=\pi^M\int d^{2M}\bm{\alpha}d^{2M}\bm{\beta} W^{(t)}_{\rho_\text{in}}(\bm{\alpha})T^{(t,s)}_{\cal E}(\bm{\beta}|\bm{\alpha})W^{(-s)}_{\Pi_{\bm{\nu}}}(\bm{\beta}) \\
&=\pi^M\int d^{2M}\bm{\alpha}d^{2M}\bm{\beta} W^{(t)}_{\rho_\text{in}}(\bm{\alpha})h(\bm{\alpha})T^{'(t,s)}_U(\bm{\beta}|\bm{\alpha})\nonumber\\
& \times g(\bm{\beta})W^{(-s)}_{\Pi_{\bm{\nu}}}(\bm{\beta}) \\
&\coloneqq \pi^M\int d^{2M}\bm{\alpha}d^{2M}\bm{\beta} W^{'(t)}_{\rho_\text{in}}(\bm{\alpha})T^{'(t,s)}_U(\bm{\beta}|\bm{\alpha})W^{'(-s)}_{\Pi_{\bm{\nu}}}(\bm{\beta}),
\end{align}
where $W^{'(t)}_{\rho_\text{in}}(\bm{\alpha})=W^{(t)}_{\rho_\text{in}}(\bm{\alpha})h(\bm{\alpha})$ and $W^{'(-s)}_{\Pi_{\bm{\nu}}}(\bm{\beta})=W^{(-s)}_{\Pi_{\bm{\nu}}}(\bm{\beta})g(\bm{\beta})$ with appropriate functions $h(\bm{\alpha})$ and $g(\bm{\beta})$. For the second equality, the transition function should satisfy a condition resulting by a symmetry in the phase space, which is given by
\begin{equation}\label{eq:condition}
T^{(t,s)}_{\cal E}(\bm{\beta}|\bm{\alpha})=h(\bm{\alpha})T^{'(t,s)}_U(\bm{\beta}|\bm{\alpha})g(\bm{\beta}),
\end{equation}
Here, the auxiliary functions $h(\bm{\alpha})$ and $g(\bm{\beta})$ have to be chosen so that the modified functions $W^{'(t)}_{\rho_\text{in}}(\bm{\alpha})$, $T^{'(t,s)}_U(\bm{\beta}|\bm{\alpha})$, and $W^{'(-s)}_{\Pi_{\bm{\nu}}}(\bm{\beta})$ are well-behaved in the phase space. An important point is that the modified functions do not need to be proper quasiprobability distributions of physical operators because we only need to exploit their shapes in the phase space. Consequently, these additional degrees of freedom allow us to further optimize the shapes of $s$-PQDs for our purposes. 

 Let us consider the linear optical setting. For $t=s$, with transition function $T_U(\bm{\beta}|\bm{\alpha})=\delta(\bm{\beta}-U\bm{\alpha})$, one can choose the modified function $T'_U(\bm{\beta}|\bm{\alpha})=T_U(\bm{\beta}|\bm{\alpha})$ with $h(\bm{\alpha})=e^{\gamma|\bm{\alpha}|^2}$, $g(\bm{\beta})=e^{-\gamma|\bm{\beta}|^2}$ with an appropriate constant $\gamma$. For this choice, we exploit the norm-preserving symmetry of linear optical transformation in the phase space, i.e., $|\bm{\alpha}|^2=|\bm{\beta}|^2$, and the product form of $s$-PQDs. As a result, the Born rule probability is rewritten as
\begin{align}\label{eq:modified}
 &p(\bm{\nu})=\pi^M\int d^{2M}\bm{\alpha}d^{2M}\bm{\beta} \prod^M_{i=1}  W_{V_i}^{(s)}(\alpha_i) \delta(\bm{\beta-U\bm{\alpha}}) \nonumber\\ &\times \prod^M_{j=1} W^{(-s)}_{\Pi_{m_j}}(\beta_j)\\ 
 &=\pi^M\int d^{2M}\bm{\alpha}d^{2M}\bm{\beta} \prod^M_{i=1}  W_{V_i}^{(s)}(\alpha_i)e^{\gamma|\bm{\alpha}|^2} \delta(\bm{\beta-U\bm{\alpha}}) \nonumber \\
 &\times e^{-\gamma|\bm{\beta}|^2}\prod^M_{j=1}  W^{(-s)}_{\Pi_{m_j}}(\beta_j) \\
 &=\pi^M\int d^{2M}\bm{\alpha} \prod^M_{i=1}  W_{V_i}^{(s)}(\alpha_i)e^{\gamma|\alpha_i|^2} \prod^M_{j=1} e^{-\gamma|\beta_j|^2} W^{(-s)}_{\Pi_{m_j}}(\beta_j) \\
 &\coloneqq \int d^{2M}\bm{\alpha} \prod^M_{i=1}  P_{i}(\alpha_i,V_i,s,\gamma)\prod^M_{j=1} f_j(\beta_j,\Pi_{m_j},s,\gamma),
\end{align}

where $P_{i}(\alpha_i,V_i,s,\gamma)=\frac{1}{{\cal N}_i}W_{V_i}^{(s)}(\alpha_i)e^{\gamma|\alpha_i|^2}$ with appropriate normalization constants ${\cal N}_i$ satisfying $\int d^2 \alpha_i P_{i}(\alpha_i,V_i,s,\gamma)=1$, and $f_j(\beta_j,\Pi_{m_j},s,\gamma)={\cal N}_je^{-\gamma|\beta_j|^2} W^{(-s)}_{\Pi_{m_j}}(\beta_j)$ with ${\cal N}_i={\cal N}_j$, for $i=j$.

\subsection{Improved approximation for outcome probability of linear optical circuit}\label{sec:res.improved}

Let us first focus on improving the approximation scheme with additive-error using our method. Since $P_{i}(\alpha_i,V_i,s,\gamma)$'s are nonnegative distributions for a Gaussian input state, the modified negativity bound is given by ${\cal M}'_{\rightarrow}=\prod^M_{j=1}\max_{\beta_j} \left|f_j(\beta_j,\Pi_{m_j},s,\gamma)\right|$. The advantage of our method is manifest especially when ${\cal M}_\rightarrow$ grows exponentially in the number of modes whereas ${\cal M}'_\rightarrow$ grows at most polynomially in the number of modes (see Fig.~\ref{fig:scheme} (b)). Let us present a simple example for which our method works well. Consider an $M$-mode identical pure squeezed vacuum states input with squeezing parameter $r>0$ and all single-photon outcomes $\bm{m}=(1,\dots,1)$. Since the negativity of the photon number measurement operator is monotonically decreasing with growing $s$, we choose $s=s_{\max}=e^{-2r}$~\cite{lee1991measure}. Then the negativity bound ${\cal M}_{\rightarrow}$ is exponential in $M$ when the squeezing is high, i.e., $r>\frac{1}{2}\log (2+\sqrt{5})$, because of $\max_{\beta_j}\left| W_{\Pi_{1}}^{(-s_{\max})}(\beta_j) \right| >1$. However, we can shift the Gaussian factor by choosing $\gamma$ as (Appendix~\ref{app:addhaf})
\begin{equation}
    \gamma^*=(1+\tanh{r})\left[ (1+\coth{r})W(1/e)-1\right],
\end{equation}
where $W(x)$ is the Lambert $W$ function. Note that $\gamma^*<0$, which means an inverse Gaussian function acts on the $s$-PQD of the measurement operator. From the fact that $\max_{\beta_j}\left| f_j(\beta_j,\Pi_{1},s_{\max},\gamma^*)\right| < 1$, the modified negativity bound ${\cal M}'_{\rightarrow}$ is exponentially small in $M$ for any squeezing $r$, which renders an efficient approximation with additive-error $\epsilon$ within running time $T=\text{poly}(M,1/\epsilon,\log\delta^{-1})$. 

Furthermore, our method provides an intriguing result on the approximation scheme with multiplicative-error (see Fig.~\ref{fig:scheme} (c)). For instance, we consider an $M$-mode identical squeezed thermal state $(r,n_{\text{th}})$ having high enough classicality $s_{\max}>1$ and all single-photon outcomes $\bm{m}=(1,\dots,1)$. In this case, an appropriate $\gamma>0$ can make $s$-PQD of the measurement operator a log-concave function by adding Gaussian smoothing.

\subsection{Application I: quantum-inspired algorithms for matrix functions}
Permanent and hafnian are important matrix functions in computational complexity. Although computing these matrix functions is generally hard~\cite{valiant1979complexity,aaronson2011linear,grier2016new,rudelson2016hafnians}, there are still efficient methods for matrices that have specific structures or restrictions~\cite{jerrum2004polynomial,barvinok2016computing,cifuentes2016efficient,oh2022classical}. 
Developing algorithms for estimating matrix functions of particular classes of matrices is a highly nontrivial problem and might enable us to understand the hardness of the problem better.

One approach is using quasiprobability representations of matrix functions. In general, there can be several ways to match the matrix functions with the outcome probability of quantum circuits (Fig.~\ref{fig:scheme} (a)). For example, an outcome probability of a linear optical circuit with photon number measurements and a Gaussian input state with zero-displacement having covariance matrix $V$ can be written using a hafnian as~\cite{hamilton2017gaussian}
\begin{equation}\label{eq:haf}
    p(\bm{m})=\frac{\Haf(A_S)}{\bm{m}!\sqrt{|V_Q|}},
\end{equation}
where $V_Q=V+\mathbb{I}_{2M}/2$, $A=\begin{pmatrix} 0 & \mathbb{I}_M \\ \mathbb{I}_M & 0\end{pmatrix}\left(\mathbb{I}_{2M}-V_Q^{-1} \right)$ and $A_S$ is a submatrix of $A$ with repeated rows and columns depending on the detected photons. Meanwhile, if the measurements are threshold detectors, i.e., $\Pi_{\text{off}}=\ket{0}\bra{0}$ and $\Pi_{\text{on}}=\mathds{1}-\ket{0}\bra{0}$, the corresponding probability is given in terms of Torontonian as~\cite{quesada2018gaussian}
\begin{equation}
p(\bm{m}')=\frac{1}{\sqrt{|V_Q|}}\text{Tor}(O_S),
\end{equation}
where $O=\mathbb{I}_{2M}-V^{-1}_Q$ and $\bm{m}'$ is an $M$-element binary vector representing on$/$off measurement outcomes. Therefore, our approximating methods for an outcome probability $p(\bm{m})$ are closely related to estimating the above matrix functions with additive or multiplicative-errors.

\subsubsection{Estimating algorithms with additive-errors}\label{sec:add}
Suppose we have a Gaussian input state as a squeezed thermal state. When the thermal part is absent, corresponding to the standard GBS circuit, we can obtain an algorithm for estimating the absolute square of the hafnian of a complex symmetric matrix:
\begin{theorem}\label{th:haf}(Estimating hafnian)
For an $M \times M$ complex symmetric matrix $R$, one can approximate $|\Haf (R)|^2$ with a success probability $1-\delta$ using the number of samples $O(\log{\delta^{-1}}/\epsilon^2)$ within the additive-error
\begin{equation}
    \epsilon \left( \frac{\lambda_{\max}}{\sqrt{1-2W(1/e)}}\right)^M\simeq \epsilon (1.502\lambda_{\max})^M,
\end{equation}
where $\lambda_{\max}$ is the largest singular value of $R$.
\end{theorem}
\begin{proof}
First, we embed the hafnian of a complex symmetric matrix to an outcome probability of a GBS circuit after rescaling the matrix so that its singular values are on the interval $[0,1)$.
More specifically, when the input is an $M$-mode product of pure squeezed states with the squeezing parameters $\{r_i\}_{i=1}^M$, the probability of a GBS circuit with all single-photon outcome $\bm{m}=(1,\dots,1)$ is $p_{\text{sq}}=\frac{1}{\cal Z}\left| \text{Haf}(R)\right|^2$ with ${\cal Z}=\prod_{i=1}^{M}\cosh{r_i}$, $R=UDU^T$, and $D=\oplus^M_{i=1} \tanh{r_i}$~\cite{hamilton2017gaussian}. Since any complex symmetric matrix can be decomposed as $UDU^T$~\cite{bunse1988singular} with a unitary matrix $U$, our algorithm can be applied to general complex symmetric matrices. 
Meanwhile, this probability can be also written by using $s$-PQDs in the form of Eq.~(\ref{eq:modified}) such as
\begin{align}
    &p_{\text{sq}}=\int d^{2M}\bm{\alpha} \prod^{M}_{i=1}\frac{1}{{\cal N}^{\text{sq}}_i}W^{(s)}_{V_{\text{sq},i}}(\alpha_i)e^{\gamma|\alpha_i|^2} \nonumber\\
    &\times \prod_{j=1}^{M} {\cal N}^{\text{sq}}_j\frac{8|\beta_j|^2+2(s^2-1)}{(s+1)^3} e^{-\left(\frac{2}{s+1}+\gamma \right)|\beta_j|^2} \nonumber \\
    &\coloneqq \int d^{2M}\bm{\alpha} \prod^M_{i=1}  P_{\text{sq},i}(\alpha_i,r_i,s,\gamma)\prod^M_{j=1} f_{\text{sq},j}(\beta_j,r_j,s,\gamma)
\end{align}
where $V_{\text{sq},i}$ is the covariance matrix of the squeezed state on the $i$th mode and ${\cal N}^{\text{sq}}_i$'s are the normalization factors for  $P_{\text{sq},i}(\alpha_i,r_i,s,\gamma)$'s. Note that $\gamma \in (-\frac{2}{s+1},\frac{2}{e^{2r_{\max}}-s})$ for a given $s$ and $r_{\max}\coloneqq \max_i r_i$.
An appropriate choice of $\gamma$~(see Appendix~\ref{app:addhaf}) gives an upper bound on $|f_{\text{sq},j}(\beta_j,r_j,\gamma,s)|$: 
\begin{align}
    \left| 
    f_{\text{sq},j}(\beta_j,r_j,\gamma,s) \right|
    &\leq \frac{\lambda_{\max}^2\sqrt{1-\lambda_j^2}}{\sqrt{\lambda_{\max}^2(1-W(1/e))^2-\lambda_j^2 W(1/e)^2}},
\end{align}
where $\lambda_j=\tanh{r_j}$, and $\lambda_{\max} \coloneqq \max_j \lambda_j$. Then by Hoeffding inequality~\cite{hoeffding1994probability},
\begin{equation}
    \text{Pr}(||\text{Haf}(R)|^2-{\cal Z}\mu|\geq {\cal Z}\epsilon) \leq 2\exp{\Big(-\frac{N\epsilon^2}{2C^{2M}}\Big)},
\end{equation}
where ${\cal Z}=\prod^M_{i=1}1/\sqrt{1-\lambda_i}$ and $C=\max_j|f_{\text{sq},j}(\beta_j,r_j,\gamma,s)|$. Given the success probability of the estimation $1-\delta$ and the number of samples $O(\log{\delta^{-1}/ \epsilon^2})$, we arrive at the result by substituting $\lambda_{j}=\lambda_{\max}$ to obtain an upper bound.
\end{proof}
Note that the above algorithm gives the finest precision, to the best of our knowledge, i.e., $\epsilon(e\lambda_{\max})^M$ in Ref.~\cite{oh2022quantum}. However, the error is still larger than what we need for the hardness conjecture in GBS~\cite{doi:10.1126/sciadv.abi7894}, so it does not lead to a contradiction. Next, if the input of the GBS circuit is a thermal state, we have an algorithm for the permanent of an HPSD matrix:

\begin{theorem}\label{th:per}(Estimating permanent of HPSD matrices)
For an $M \times M$ HPSD matrix $B$, one can approximate $\Per (B)$ with a success probability $1-\delta$ using the number of samples $O(\log{\delta^{-1}}/\epsilon^2)$ within the error
\begin{equation}
    \epsilon \prod_{i=1}^M \frac{4\lambda_{\max}^2}{e(2\lambda_{\max}-\lambda_i)},
\end{equation}
where $\lambda_i$ are singular values of the matrix $B$ and $\lambda_{\max}$ is the largest one.
\end{theorem}
\begin{proof}
This corresponds to the case where the input is an $M$-mode product of thermal states with the average photon numbers $\{n_i\}_{i=1}^M$. The probability of all single-photon outcomes matches the permanent of an HPSD matrix $B$ such as $p_{\text{th}}=\frac{1}{\cal Z'}\Per(B),
$ where ${\cal Z'}=\prod_{i=1}^M (1+n_i)$, $B=UDU^{\dagger}$, and $D=\text{diag}\{ \frac{n_1}{n_1+1},\dots,\frac{n_M}{n_M+1}\}$. Then we obtain the result by a similar procedure in the hafnian case. A detailed derivation is given in Appendix~\ref{app:addper}.
\end{proof}
The precision of our result outperforms the best-known one~\cite{chakhmakhchyan2017quantum} because the latter is a special case of ours when $s=1$ and $\gamma=0$~(see Fig.~\ref{fig:scheme} (b)). Specifically, when $\lambda_{\max} \in (0,1/2)$, our method's precision is better than the previous one (see Appendix~\ref{app:addper}).  Moreover, when the input is a squeezed thermal state, we obtain an algorithm for the hafnian of a structured matrix (Appendix~\ref{app:addhaf}). Similarly, we provide algorithms for the Torontonian of some structured matrices within additive-error by substituting the photon number measurement by a threshold detector. The detailed results are in Appendix~\ref{app:addtor}.

\subsubsection{Estimating algorithms with multiplicative-errors}\label{sec:mul}
Recall that we have an FPRAS when the estimate function is log-concave. 
Thus our goal is to make the estimate function log-concave by controlling the parameters $s$ and $\gamma$ in our scheme. 
For the permanent of an HPSD matrix, this is possible when $\lambda_{\max}/\lambda_{\min}\leq 2$ (a proof in Appendix~\ref{app:mulper}), which reproduces the existing result of Ref.~\cite{barvinok2021remark}
in the case $1\leq \lambda_i \leq 2$ after a normalization. Especially for an HPSD matrix with $\lambda_{\min}=0$, estimating the permanent within a multiplicative-error is NP-hard~\cite{meiburg2021inapproximability}; thus $\lambda_i>0$ is the crucial condition for the efficient approximation.    
We emphasize that our formulation comes from a physical setup, which is essentially different from the method in Ref.~\cite{barvinok2021remark}, where the technique is restricted to the properties of the permanent of a positive definite matrix. Thus our result can be readily extended to a more general situation other than the permanent of a positive definite matrix. When the input state is a product of squeezed thermal states $\{r_i,n_i\}_{i=1}^M$, we have an FPRAS for the hafnian of a matrix having a specific form, such that $2 \times 2$ block matrix whose diagonal elements are symmetric matrices and off-diagonal elements are HPSD matrices. 
\begin{theorem}\label{th:FPRAS}(FPRAS for hafnian)
Suppose we have a block matrix $A=\begin{pmatrix} R & B \\ B^T & R^*\end{pmatrix}$ with an $M \times M$ complex symmetric matrix $R$ and an $M \times M$ HPSD matrix $B$, which have decompositions by a unitary matrix $U$ as $UDU^T$ and $UD'U^{\dagger}$, respectively, with
\begin{align}\label{eq:acondition1}
D&=\bigoplus^M_{i=1}\frac{(1+2n)\sinh{2r_i}}{1+2n(1+n)+(1+2n)\cosh{2r_i}},\\ 
    D'&=\bigoplus^M_{i=1}\frac{2n(1+n)}{1+2n(1+n)+(1+2n)\cosh{2r_i}},\label{eq:acondition2}
\end{align}
where $n=n_i$ for all $i$ and $n, r_i \geq 0$. Then $\Haf(A)$ can be approximated by FPRAS when the parameters satisfy a condition as 
\begin{equation}
    n \geq \frac{1}{4} \left(6 \sinh (2 r_{\max})+\sqrt{18 \cosh (4 r_{\max})-14}-2\right),
\end{equation}
where $r_{\max}=\max_i r_i$.
\end{theorem}
\begin{proof}
See Appendix~\ref{app:FPRAShaf}.
\end{proof}

We also consider an on/off measurement instead of the photon number measurement, which corresponds to the Torontonian (detailed analysis is given in Appendix~\ref{app:FPRAStor}). 

Moreover, we give nontrivial lower and upper bounds on the values of various matrix functions including the permanent and hafnian by adjusting the parameter $s$, which have independent interests (Appendix~\ref{app:bounds})~\cite{gurvits2014bounds}. 

\subsection{Application II: Sparse Gaussian boson sampling}\label{sec:sparse}
Our algorithm has an interesting application to the simulation of GBS. Recall that if the input is a classical state, the corresponding GBS can be efficiently simulated ~\cite{rahimi2016sufficient}. In our language, this is the case when $s_{\max}\geq 1$. For a non-classical input state ($s_{\max} < 1$), however, a classically efficient simulation may not be possible from the hardness of GBS. Nevertheless, we can approximate any (marginal) outcome probability of the circuit using our algorithm under certain conditions. Then one can simulate the GBS when its output distribution is poly-sparse, in which a probability distribution with polynomially many of the most likely outcomes well approximates the true distribution~\cite{schwarz2013simulating,pashayan2020estimation}.
\begin{theorem}\label{th:GBS}(Estimating outcome probabilities of GBS) For a lossy GBS circuit with squeezing  $\{r_i\}_{i=1}^M$  and a transmissivity $\eta$, one can efficiently approximate any (marginal) outcome probability within $1/\text{poly}(M)$ additive-error when $\eta e^{-2r_{\max}}+1-\eta \geq \sqrt{5}-2$ with $r_{\max}=\max_i r_i$.
\end{theorem}

\begin{proof}
Let us first consider a GBS circuit with a product of lossy squeezed input states with a transmissivity $\eta$ having the covariance matrix on $i$th mode as $V_{\eta,i}=\frac{1}{2}\begin{pmatrix}\eta e^{2r_i}+1-\eta & 0 \\ 0 & \eta e^{-2r_i}+1-\eta  \end{pmatrix}$ and a photon number measurement $\Pi_{m_j}=\ket{m_j}\bra{m_j}$. Then an outcome probability of $\bm{m}=(m_1,...,m_M)$ is given by
\begin{align}\label{eq:GBS}
    p(\bm{m})&=\pi^M\int d^{2M}\bm{\alpha} \prod^M_{i=1}  W_{V_{\eta,i}}^{(s)}(\alpha_i)\prod^M_{j=1} W^{(-s)}_{\Pi_{m_j}}(\beta_j).
\end{align}
We take $s=s_{\max}=\eta e^{-2r_{\max}}+1-\eta$. If we examine the probability of all single-photon outcomes $\bm{m}=(1,\dots,1)$,
a condition for an efficient estimation is $\max_{\beta_j}|\pi W^{(-s_{\max})}_{\Pi_1}(\beta_j)|\leq 1$, which leads the restriction $s_{\max}\geq \sqrt{5}-2$. Now we must check whether this condition is valid for other outcomes. From the behavior of $W^{(-s)}_{\Pi_m}(\beta)$, we can find out that 
\begin{equation}
    \max_{\beta} |W^{(-s)}_{\Pi_m}(\beta)| \leq \max_{\beta} |W^{(-s)}_{\Pi_1}(\beta)|,
\end{equation}
for $n \geq 2$ and $s\geq 0$. Lastly,
we consider $\pi W^{(-s)}_{\Pi_0}(\beta_j)$ for zero-photon detection and $\sum_{m_j} \pi W^{(-s)}_{\Pi_{m_j}}(\beta_j)$ for the marginalized probability on mode $j$, where the latter sum equals to one by the normalization condition.
In both cases, the integrals for $\beta_j$'s can be easily computed because $\beta_j$ components in $W^{(s)}_{V_{\eta}}(\bm{\alpha})$ constitute a multivariate normal distribution. Therefore, we can first perform the integrals on $\beta_j$'s corresponding to zero-photon or marginalized ones and estimate the remaining terms.
\end{proof}
Note that the condition in Theorem~\ref{th:GBS} yields  $r_{\max}\leq \frac{1}{2}\log(2+\sqrt{5})\simeq 0.722$ for an ideal GBS ($\eta=1$). However, under photon loss, any squeezed input state is possible when $\eta \leq 3-\sqrt{5}\simeq 0.764$, which is much higher transmissivity than those used in current experiments~\cite{zhong2021phase,madsen2022quantum}. 

We also consider GBS with threshold detectors~\cite{quesada2018gaussian}, whose $s$-PQD for ``click'' is given by
\begin{equation}
\pi W^{(-s)}_{\text{on}}(\beta)= 1-\frac{2}{s+1}e^{-\frac{2}{s+1}|\beta|^2}.
\end{equation}
Since the range of $\pi W^{(0)}_{\text{on}}(\beta)=1-2e^{2|\beta|^2}$ is  $[-1,1]$ in the Wigner representation ($s=0$), it allows any squeezing of the input state even for the lossless case ($\eta=1$). 
Then we can estimate any (marginal) probability by the same argument of the photon-number measurement case. We give a detailed analysis in Appendix~\ref{app:GBSsimul}. 

In both cases, those algorithms of estimating (marginal) probability with inverse-polynomial additive-error precision with the poly-sparsity condition imply classically efficient simulations of GBS. Since it is difficult to expect the sparsity condition for current experiments, our results do not lead to a direct simulation of them. However, it can be useful for application-targeted GBS experiments whose outcomes might
satisfy the poly-sparsity condition.

Furthermore, we might consider an additional source of noise to allow $s_{\max}>1$ by introducing thermal noise with the average photon number $n_{\text{th}}$. In that case, the covariance matrix of $i$th mode input state is  $V_{\eta,n_{\text{th}},i}=\frac{1}{2}\begin{pmatrix}\eta e^{2r_i}+(1-\eta)(2n_{\text{th}}+1) & 0 \\ 0 & \eta e^{-2r_i}+(1-\eta)(2n_{\text{th}}+1)  \end{pmatrix}$.
By the log concavity condition, we can compute any (marginal) probability within a multiplicative-error if the following condition is satisfied:
\begin{equation}\label{eq:crit}
    n_{\text{th}} \geq \frac{e^{-r_{\max}}\eta\sinh{r_{\max}}+\sqrt{1+\eta\sinh{2r_{\max}}}}{1-\eta}>1.
\end{equation}
For instance, if $\eta=0.5$ and $r_{\max}=1$, then $n_{\text{th}} \geq n_{\text{th}}^* \simeq 3.79$, and the minimum value of $n_{\text{th}}^*$ is $1$ when $\eta \rightarrow 0$. 
 
In Fig.~\ref{fig:precision}, we depict our results of approximating probability and simulation for GBS circuits with threshold detectors via the classicality $s_{\max}$. Remarkably, there are two transition points of complexities via the classicality $s_{\text{max}}$. i) $s_{\max}=1$: This point indicates the transition from $\epsilon$-simulation with the sparse condition to exact simulation. Whether $\epsilon$-simulation without sparse condition can exist when $s_{\text{max}}<1$ is an interesting open question. ii) $s_{\text{max}}=s_{\text{m}}$: Note that $s_{\text{m}}$ is the point saturating the inequality Eq.~(\ref{eq:crit}), which indicates the transition of complexities from $\text{BPP}^{\text{NP}}$ to $\text{BPP}$ for approximating the probability within a multiplicative error. 

\begin{figure}[t]
\includegraphics[width=240px]{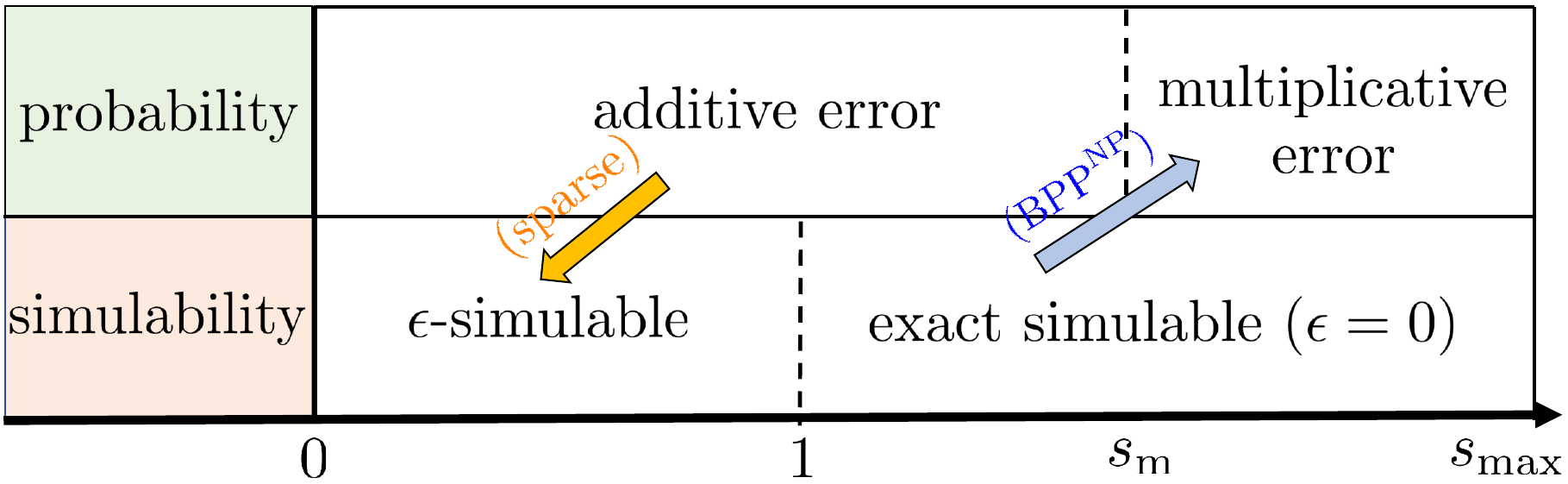}
\caption{Regime of efficient classical algorithms for approximating outcome probabilities and the simulation of a GBS circuit with threshold detectors via the classicality $s_{\max}$. When the output distribution is poly-sparse, an approximate simulation is possible by estimating the probabilities within 1/poly additive-error (orange arrow) ~\cite{schwarz2013simulating}. On the other hand, a multiplicative-error approximation of probability in $\text{BPP}^{\text{NP}}$ is possible by using exact simulation with classical input state $s_{\text{max}}\geq 1$ (blue arrow)~\cite{rahimi2015can}.} 
\label{fig:precision}
\end{figure}

\section{Discussion}
We propose a method for calculating the outcome probability of a linear optical
circuit, by introducing modified functions with a lower negativity bound than the quasi-probability
distribution. This leads to various improved approximating algorithms for the outcome probabilities of the circuit. Furthermore, we suggest an FPRAS using our method modulating $s$-PQDs and the efficient sampling of log-concave functions with multiplicative-error.
Our results provide a helpful tool for controlling the negativity of the circuit in the phase space and interesting quantum-inspired algorithms in computational complexity.

Although we focus on Gaussian input states and photon number or threshold measurements in a linear optical circuit in this work, our scheme can be also applied to other systems, for example, Clifford circuits~\cite{jozsa2013classical} with a dimension of odd prime $d$. If one exploits the symmetry of the transition function in the phase space, a nontrivial approximation algorithm can be rendered for the corresponding matrix function. Since there can be several equivalent choices of circuits for the same matrix function, e.g., permanent of a unitary matrix~\cite{grier2021complexity,chabaud2022quantum}, finding the optimal circuit is still an interesting open problem. After a proper circuit choice, there are other optimization problems, such as choosing quasiprobability representation and optimizing parameters for manipulating them in the phase space. Therefore, there might be more improvements for quantum-inspired algorithms approximating matrix functions.

\section*{data availability}
The data supporting the results of this manuscript are given in the article and the appendix. Extra data are available upon reasonable request.

\section*{acknowledgement}
We thank Abhinav Deshpande and Bill Fefferman for fruitful discussions. Y. L. acknowledges National Research Foundation of Korea a grant funded by the Ministry of Science and ICT (NRF-2020M3E4A1077861), Institute of Information \& Communications Technology Planning \& Evaluation (IITP) grant funded by the Korea government (MSIT) (No. 2022-0-00463), and KIAS Individual Grant
(CG073301) at Korea Institute for Advanced Study. 
C. O. acknowledges support from the ARO MURI (W911NF-21-1-0325) and
NSF (OMA-1936118, ERC-1941583).

\section*{Author Contribution}
Y. L. initiated the idea and Y. L. and C. O. led and analyzed the main results. All authors wrote and reviewed the manuscript.

\section*{Competing interests}
The authors declare no competing interests.

\begin{widetext}
\appendix

\section{Estimation of matrix functions within additive-errors}\label{app:additive}
\subsection{Hafnian (Proof of Theorem~\ref{th:haf})}\label{app:addhaf}
Consider an input of $M$-mode product of pure squeezed vacuum states $\{r_i\}_{i=1}^M$ and the all single-photon outcomes $\bm{m}=(1,...,1)$. From Eq.~(\ref{eq:haf}), 
\begin{equation}\label{eq:haf2}
    p_{\text{sq}}=\frac{1}{\cal Z}\left| \text{Haf}(R')\right|^2,
\end{equation}
where ${\cal Z}=\prod_{i=1}^{M}\cosh{r_i}$, $R'=UDU^T$, and $D=\oplus^M_{i=1} \tanh{r_i}$. Since any complex symmetric matrix can be decomposed by Takagi decomposition as $UDU^T$~\cite{bunse1988singular}, the only restriction is the magnitude of singular values $\lambda_i=\tanh{r_i} \in [0,1)$.

For a given complex symmetric matrix $R$, we can construct a quantum circuit, the probability of which is expressed as its hafnian. To do that, first rescale the matrix with the largest singular value $\lambda_{\max}$ as $R'=R/(a\lambda_{\max})$ with $a>1$, and find the Takagi decomposition of $R'$ as $R'=UDU^T$ so that a GBS probability is matched to $|\text{Haf}(R')|^2$. From Eq.~(\ref{eq:haf2}),
\begin{equation}
    |\text{Haf}(R)|^2=(a\lambda_{\max})^{M}|\text{Haf}(R')|^2=(a\lambda_{\max})^{M}{\cal Z}p_{\text{sq}}.
\end{equation}
If an estimator of $|\text{Haf}(R)|^2$ lies in the interval $[-C^M,C^M]$, by Hoeffding inequality~\cite{hoeffding1994probability},
\begin{equation}\label{eq:hafestimation}
\text{Pr}(||\text{Haf}(R)|^2-(a\lambda_{\max})^{M}{\cal Z}\mu|\geq (a\lambda_{\max})^{M}{\cal Z}\epsilon) \leq 2\exp{\Big(-\frac{N\epsilon^2}{2C^{2M}}\Big)},
\end{equation}
where $\mu$ is the sample mean of $p_{\text{sq}}$. With a success probability $1-\delta$, a sufficient number of samples for the estimation of $|\text{Haf}(R)|^2$ within an additive-error $\epsilon$ is given by
\begin{equation}
    N=\frac{2(a\lambda_{\max})^{2M}{\cal Z}^2C^{2M}}{\epsilon^2}\ln\frac{2}{\delta}.
\end{equation}

In other words, if we fix the sample size as $N=O(\log{\delta^{-1}}/\epsilon^2)$, one can estimate $|\text{Haf}(R)|^2$ with a success probability $1-\delta$ within an additive-error $2\epsilon (a\lambda_{\max}C{\cal Z}^{1/M})^M$. To obtain the bound $C$, we express the probability using $s$-PQDs as 
\begin{align}
     p_{\text{sq}}
    &=\int d^{2M}\bm{\alpha}\prod^{M}_{i=1}\frac{1}{\pi \sqrt{\det{(V_{\text{sq},i}-s/2)}}}e^{-\bm{\alpha}_i {(V_{\text{sq},i}-s/2)}^{-1}\bm{\alpha}_i^T}\prod_{j=1}^{M} \frac{8|\beta_j|^2+2(s^2-1)}{(s+1)^3} e^{-\frac{2|\beta_j|^2}{s+1}}
    \\ &=\int d^{2M}\bm{\alpha} \prod^{M}_{i=1}\frac{2}{\pi \sqrt{(e^{2r_i}-s)(e^{-2r_i}-s)}}e^{-\frac{2\alpha_{ix}^2}{e^{2r_i}-s}-\frac{2\alpha_{iy}^2}{e^{-2r_i}-s}}\prod_{j=1}^{M} \frac{8|\beta_j|^2+2(s^2-1)}{(s+1)^3} e^{-\frac{2|\beta_j|^2}{s+1}}
    \\ &=\int d^{2M}\bm{\alpha} \prod^{M}_{i=1}\frac{2}{\pi \sqrt{(e^{2r_i}-s)(e^{-2r_i}-s)}}e^{-\left(\frac{2}{e^{2r_i}-s}-\gamma \frac{2}{e^{2r_{\max}}-s} \right)\alpha_{ix}^2-\left(\frac{2}{e^{-2r_i}-s}-\gamma \frac{2}{e^{2r_{\max}}-s} \right)\alpha_{iy}^2} \nonumber\\
    &\times \prod_{j=1}^{M} \frac{8|\beta_j|^2+2(s^2-1)}{(s+1)^3} e^{-\left(\frac{2}{s+1}+\gamma \frac{2}{e^{2r_{\max}}-s}\right)|\beta_j|^2}
    \\&=\int d^{2M}\bm{\alpha} \prod^{M}_{i=1}\frac{2}{\pi} \sqrt{\frac{e^{2r_{\max}}-s-\gamma(e^{2r_i}-s)}{(e^{2r_i}-s)(e^{2r_{\max}}-s)}}\sqrt{\frac{e^{2r_{\max}}-s-\gamma(e^{-2r_i}-s)}{(e^{-2r_i}-s)(e^{2r_{\max}}-s)}}e^{-\left(\frac{2}{e^{2r_i}-s}-\gamma \frac{2}{e^{2r_{\max}}-s} \right)\alpha_{ix}^2-\left(\frac{2}{e^{-2r_i}-s}-\gamma \frac{2}{e^{2r_{\max}}-s} \right)\alpha_{iy}^2}\nonumber \\
    &\times \prod_{j=1}^{M} \sqrt{\frac{e^{2r_{\max}}-s}{e^{2r_{\max}}-s-\gamma(e^{2r_j}-s)}}\sqrt{\frac{e^{2r_{\max}}-s}{e^{2r_{\max}}-s-\gamma(e^{-2r_j}-s)}}\frac{8|\beta_j|^2+2(s^2-1)}{(s+1)^3} e^{-\left(\frac{2}{s+1}+\gamma \frac{2}{e^{2r_{\max}}-s}\right)|\beta_j|^2}\\
    &\coloneqq \int d^{2M}\bm{\alpha}\prod^M_{i=1}P_{\text{sq},i}(\alpha_i,r_i,\gamma,s)\prod^M_{j=1}f_{\text{sq},j}(\beta_j,r_j,\gamma,s),\label{eq:sq}
\end{align}
where $V_{\text{sq},i}$ is the covariance matrix of a squeezed vacuum state in mode $i$ and $\gamma \in [0, 1)$ is the (normalized) parameter shifting the Gaussian factor such that $\gamma \rightarrow 1$ ($\gamma = 0$) means maximum (no) shifting. To obtain a bound on $|f_{\text{sq},j}(\beta_j,r_j,\gamma,s)|$, let us set $s=s_{\max}=e^{-2r_{\max}}$ and $r_j=\text{tanh}^{-1}\lambda_j$. The extreme points of $f_{\text{sq},j}$ are $0,\pm \beta^*$ with
\begin{equation}
    \beta^*=\sqrt{\frac{\lambda_{\max} (\gamma  (\lambda_{\max}-1)-4 \lambda_{\max}-2)}{(\lambda_{\max}+1)^2 (\gamma  (\lambda_{\max}-1)-2 \lambda_{\max})}}.
\end{equation}
Note that for $s_{\max} < 1$, $f_{\text{sq},j}(0,\lambda_j,\gamma,s_{\max}) <0$ and $f_{\text{sq},j}(\beta^*,\lambda_j,\gamma,s_{\max})>0$. Since the extreme values are changing monotonically as $\gamma$, we choose $\gamma$ satisfying the condition $-f_{\text{sq},j}(0,\lambda_j,\gamma,s_{\max})=f_{\text{sq},j}(\beta^*,\lambda_j,\gamma,s_{\max})$, which is given by
\begin{equation}
    \gamma^*=\frac{2(1+\lambda_{\max})W(1/e)-2\lambda_{\max}}{1-\lambda_{\max}} ~~\text{  for  } 0 \leq \lambda_{\max} < \frac{W(1/e)}{1-W(1/e)},
\end{equation}
where $W(x)$ is the Lambert $W$ function, and $\frac{W(1/e)}{1-W(1/e)} \simeq 0.386$. Then an upper bound is obtained by substituting $\gamma^*$ into $f_{\text{sq},j}(\beta^*,\lambda_j,\gamma,s_{\max})$, such that
\begin{align}
    \min_{s,\gamma}\max_{\beta_j}& \left| 
    f_{\text{sq},j}(\beta_j,\lambda_j,\gamma,s) \right|
    \leq \frac{\lambda_{\max}^2\sqrt{1-\lambda_j^2}}{\sqrt{\lambda_{\max}^2(1-W(1/e))^2-\lambda_j^2 W(1/e)^2}} ~~\text{  for  } 0 \leq \lambda_{\max} < \frac{W(1/e)}{1-W(1/e)},
\end{align}

Although this bound is valid only for a certain range of $\lambda_{\max}$, we can also find the same bound out of the range by shifting the Gaussian factor in the reverse direction. Specifically,
\begin{align}
    p_{\text{sq}}&=\int d^{2M}\bm{\alpha} \prod^{M}_{i=1}\frac{2}{\pi \sqrt{(e^{2r_i}-s)(e^{-2r_i}-s)}}e^{-\frac{2\alpha_{ix}^2}{e^{2r_i}-s}-\frac{2\alpha_{iy}^2}{e^{-2r_i}-s}}\prod_{j=1}^{M} \frac{8|\beta_j|^2+2(s^2-1)}{(s+1)^3} e^{-\frac{2|\beta_j|^2}{s+1}} \\
    &=\int d^{2M}\bm{\alpha} \prod^{M}_{i=1}\frac{2}{\pi \sqrt{(e^{2r_i}-s)(e^{-2r_i}-s)}}e^{-\left(\frac{2}{e^{2r_i}-s}+\gamma'\frac{2}{s+1}\right)\alpha_{ix}^2}e^{-\left(\frac{2}{e^{-2r_i}-s}+\gamma'\frac{2}{s+1}\right)\alpha_{iy}^2}\prod_{j=1}^{M} \frac{8|\beta_j|^2+2(s^2-1)}{(s+1)^3} e^{-(1-\gamma')\frac{2|\beta_j|^2}{s+1}}\\
    &=\int d^{2M}\bm{\alpha} \prod^{M}_{i=1}\frac{2}{\pi}\sqrt{\frac{s+1+\gamma'(e^{2r_i}-s)}{(e^{2r_i}-s)(s+1)}}\sqrt{\frac{s+1+\gamma'(e^{-2r_i}-s)}{(e^{-2r_i}-s)(s+1)}}e^{-\left(\frac{2}{e^{2r_i}-s}+\gamma'\frac{2}{s+1}\right)\alpha_{ix}^2}e^{-\left(\frac{2}{e^{-2r_i}-s}+\gamma'\frac{2}{s+1}\right)\alpha_{iy}^2}\nonumber \\
    & \times\prod_{j=1}^{M}\sqrt{\frac{s+1}{s+1+\gamma'(e^{2r_j}-s)}}\sqrt{\frac{s+1}{s+1+\gamma'(e^{-2r_j}-s)}} \frac{8|\beta_j|^2+2(s^2-1)}{(s+1)^3} e^{-(1-\gamma')\frac{2|\beta_j|^2}{s+1}}\\
    &\coloneqq \int d^{2M}\bm{\alpha}\prod^M_{i=1}P'_{\text{sq},i}(\alpha_i,r_i,\gamma',s)\prod^M_{j=1}f'_{\text{sq},j}(\beta_j,r_j,\gamma',s),
\end{align}
where $\gamma' \in [0,1)$ is the parameter shifting the exponential term in the reverse direction such that $\gamma' \rightarrow 1$ $(\gamma'=0)$ means maximum (no) shifting. Similarly, the extreme points are $0,\pm \beta'^*$ with
\begin{equation}
    \beta'^*=\sqrt{\frac{(\gamma' -2) \lambda_{\max}-1}{(\gamma' -1) (\lambda_{\max}+1)^2}}.
\end{equation}
Then the $\gamma'$ satisfying the condition $-f'_{\text{sq},j}(0,\lambda_j,\gamma',s_{\max})=f_{\text{sq},j}(\beta'^*,\lambda_j,\gamma',s_{\max})$ is given by
\begin{equation}
    \gamma'^*=\frac{\lambda_{\max}-(1+\lambda_{\max})W(1/e)}{\lambda_{\max}}~~\text{  for  } \frac{W(1/e)}{1-W(1/e)} \leq \lambda_{\max} < 1.
\end{equation}
Finally we obtain an upper bound such that
\begin{align}
    \min_{s,\gamma'}\max_{\beta_j}& \left| 
    f'_j(\beta_j,\lambda_j,\gamma',s) \right|
    \leq \frac{\lambda_{\max}^2\sqrt{1-\lambda_j^2}}{\sqrt{\lambda_{\max}^2(1-W(1/e))^2-\lambda_j^2 W(1/e)^2}} ~~\text{  for  } \frac{W(1/e)}{1-W(1/e)} \leq \lambda_{\max} < 1,
\end{align}
which covers the remaining range of $\lambda_{\max}$.

Note that
\begin{equation}
    {\cal Z}=\prod^M_{i=1}\cosh{r_i}=\frac{1}{\prod^M_{i=1}\sqrt{1-\lambda_i^2}}.
\end{equation}
Consequently, by Eq.~(\ref{eq:hafestimation}), for a given complex $M \times M$ matrix $R$, and for the number of samples $N=O(\log{\delta^{-1}}/\epsilon^2)$, we can estimate $\left| \Haf(R)\right|^2$ with a success probability $1-\delta$ within an additive-error as
\begin{equation}\label{eq:addhafr}
    \epsilon \prod^M_{i=1}\frac{a\lambda_{\max}\lambda^{'2}_{\max}}{\sqrt{\lambda^{'2}_{\max}(W(1/e)-1)^2-\lambda^{'2}_i W(1/e)^2}}=\epsilon \prod^M_{i=1} \frac{\lambda_{\max}^2}{\sqrt{\lambda^{2}_{\max}(W(1/e)-1)^2-\lambda^{2}_i W(1/e)^2}},
\end{equation}
where $\lambda'_i=\frac{\lambda_i}{a\lambda_{\max}}$. 
We find upper and lower bound by setting $\lambda_i=\lambda_{\max}$ and $\lambda_i=0$, respectively, such as
\begin{equation}
   \epsilon (1.386\lambda_{\max})^M \simeq\left( \epsilon  \frac{\lambda_{\max}}{1-W(1/e)}\right)^M \leq \epsilon \prod^M_{i=1} \frac{\lambda_{\max}^2}{\sqrt{\lambda^{2}_{\max}(W(1/e)-1)^2-\lambda^{2}_i W(1/e)^2}}\leq \epsilon \left( \frac{\lambda_{\max}}{\sqrt{1-2W(1/e)}}\right)^M \simeq \epsilon (1.502\lambda_{\max})^M.
\end{equation}

Next, we consider a squeezed thermal input state $\{r_i,n\}_{i=1}^M$ to allow $s_{\max}\geq 1$. For simplicity, the average thermal photon number $n$ is fixed. Then the probability of all single-photon outcome is written by hafnian as
\begin{equation}
    p_{\text{st}}=\frac{1}{\sqrt{|V_Q^{\text{st}}|}}\Haf (A),
\end{equation}
where $V_Q^{\text{st}}=V_{\text{st}}+\mathbb{I}_{2M}/2$ with $V_{\text{st}}$ the covariance matrix of a squeezed thermal state and $A=\begin{pmatrix} R & B \\ B^T & R^*\end{pmatrix}$ with a symmetric matrix $R$ and an HPSD matrix $B$ decomposed as
Eqs.~(\ref{eq:acondition1}) and (\ref{eq:acondition2}).
Let us define $a_{\pm}(r_i,n)=(2n+1)e^{\pm 2r_i}$, $a_{\min}=(2n+1)e^{-2r_{\max}}$, and $a_{\max}=(2n+1)e^{2r_{\max}}$.
Then the probability $p_{\text{st}}$ is written by $s$-PQDs as
\begin{align}
   &p_{\text{st}}=\int d^{2M}\bm{\alpha} \prod^{M}_{i=1}\frac{2}{\pi \sqrt{(a_+(r_i,n)-s)(a_-(r_i,n)-s)}}e^{-\frac{2\alpha_{ix}^2}{a_+(r_i,n)-s}-\frac{2\alpha_{iy}^2}{a_-(r_i,n)-s}} \prod_{j=1}^{M} \frac{8|\beta_j|^2+2(s^2-1)}{(s+1)^3} e^{-\frac{2|\beta_j|^2}{s+1}} \\
   &=\int d^{2M}\bm{\alpha} \prod^{M}_{i=1}\frac{2}{\pi \sqrt{(a_+(r_i,n)-s)(a_-(r_i,n)-s)}}e^{-\left(\frac{2}{a_+(r_i,n)-s}+\gamma'\frac{2}{s+1}\right)\alpha_{ix}^2}e^{-\left(\frac{2}{a_-(r_i,n)-s}+\gamma'\frac{2}{s+1}\right)\alpha_{iy}^2} \nonumber \\
   &\times \prod_{j=1}^{M} \frac{8|\beta_j|^2+2(s^2-1)}{(s+1)^3} e^{-(1-\gamma')\frac{2|\beta_j|^2}{s+1}}\\
    &=\int d^{2M}\bm{\alpha} \prod^{M}_{i=1}\frac{2}{\pi}\sqrt{\frac{s+1+\gamma'(a_+(r_i,n)-s)}{(a_+(r_i,n)-s)(s+1)}}\sqrt{\frac{s+1+\gamma'(a_-(r_i,n)-s)}{(a_-(r_i,n)-s)(s+1)}}e^{-\left(\frac{2}{a_+(r_i,n)-s}+\gamma'\frac{2}{s+1}\right)\alpha_{ix}^2}e^{-\left(\frac{2}{a_-(r_i,n)-s}+\gamma'\frac{2}{s+1}\right)\alpha_{iy}^2}\nonumber \\
    & \times\prod_{j=1}^{M}\sqrt{\frac{s+1}{s+1+\gamma'(a_+(r_j,n)-s)}}\sqrt{\frac{s+1}{s+1+\gamma'(a_-(r_j,n)-s)}} \frac{8|\beta_j|^2+2(s^2-1)}{(s+1)^3} e^{-(1-\gamma')\frac{2|\beta_j|^2}{s+1}}\\
    &\coloneqq \int d^{2M}\bm{\alpha}\prod^M_{i=1}P'_{\text{st},i}(\alpha_i,r_i,n,\gamma',s)\prod^M_{j=1}f'_{\text{st},j}(\beta_j,r_j,n,\gamma',s),
\end{align}
where we use the reverse shifting of the Gaussian factor. To obtain an upper bound on $|f'_{\text{st},j}(\beta_j,r_j,n,\gamma',s)|$, let us set $s=s_{\max}=a_{\min}$. The extreme points are $0,\pm \beta'^*$ with
\begin{equation}
    \beta'^*=\frac{1}{2} \sqrt{\frac{e^{-4 r_{\max}} \left(2 n+e^{2 r_{\max}}+1\right) \left\{(\gamma' -3) e^{2 r_{\max}}-(\gamma' -1) (2 n+1)\right\}}{\gamma' -1}}.
\end{equation}
After choosing $\gamma'^*=e^{-\tanh{r_{\max}}}\frac{n}{n+1}$,
\begin{align}
    &\min_{s,\gamma'}\max_{\beta_j}|f'_{\text{st},j}(\beta_j,r_j,n,\gamma',s)|\leq |f'_{\text{st},j}(\beta'^*,r_j,n,\gamma'^*,a_{\min})|\\
    &=\exp\left[ \frac{1}{2}\left\{ -3+\frac{ne^{-\tanh{r_{\max}}}}{1+n}-(2n+1)e^{-2r_{\max}}\left(-1+\frac{ne^{-\tanh{r_{\max}}}}{1+n}\right)+2r_j+8r_{\max}+3\tanh{r_{\max}}\right\}\right] \nonumber \\
    &\times \frac{4(1+n)^2}{(1+e^{2r_{\max}}+2n)\{(1+n)e^{\tanh{r_{\max}}}-n\}}\sqrt{\frac{1}{1+e^{2r_{\max}}+3n+ne^{2r_{\max}}+2n^2+n(2n+1)(e^{2(r_j+r_{\max})}-1)e^{-\tanh{r_{\max}}}}} \nonumber\\
    &\times \sqrt{\frac{1}{(2n+1)e^{2r_j}\left\{-n+(n+1)e^{\tanh{r_{\max}}} \right\}+e^{2r_{\max}}\left\{ n+2n^2+(n+1)e^{2r_j+\tanh{r_{\max}}}\right\}}}.
\end{align}
Meanwhile, 
\begin{equation}
    \sqrt{|V^{\text{st}}_Q|}=\prod_{i=1}^M \sqrt{\frac{1}{2}+n(n+1)+(n+\frac{1}{2})\cosh{2r_i}}.
\end{equation}
Finally, for a matrix $A$ satisfying Eqs.~(\ref{eq:acondition1}) and (\ref{eq:acondition2}), we can estimate $\Haf(A)$ with a success probability $1-\delta$ using number of samples $N=O(\log{\delta^{-1}}/\epsilon^2)$ within the additive-error given by
\begin{equation}\label{eq:addhafa}
    \epsilon \prod_{i=1}^M \sqrt{\frac{1}{2}+n(n+1)+(n+\frac{1}{2})\cosh{2r_i}}~f'_{\text{st},i}(\beta'^*,r_i,n,\gamma'^*,a_{\min})\coloneqq \epsilon \prod_{i=1}^M H_i^A(n,r_i).
\end{equation}

\subsection{Permanent of Hermitian positive semidefinite matrices (Proof of Theorem~\ref{th:per})}\label{app:addper}
When a thermal state input with average photon numbers $\{n_i\}_{i=1}^M$ goes through a linear optical circuit instead of a squeezed vacuum state, the probability of all single-photon outcomes corresponds to the permanent of HPSD matrices~\cite{chakhmakhchyan2017quantum}. In Ref.~\cite{chakhmakhchyan2017quantum}, an algorithm for estimating the permanent of an HPSD matrix is proposed. Here, we improve the precision of the estimation using $s$-PQDs and shifting Gaussian factors. The probability of all single-photon outcomes is connected to the permanent of an HPSD matrix such that~\cite{chakhmakhchyan2017quantum}
\begin{equation}\label{eq:pergbs}
    p_{\text{th}}=\frac{1}{\cal Z'}\Per(B'),
\end{equation}
where ${\cal Z'}=\prod_{i=1}^M (1+n_i)$, $B'=UDU^{\dagger}$, and $D=\text{diag}\{ \frac{n_1}{n_1+1},...,\frac{n_M}{n_M+1}\}$. Thus if we have an $M \times M$ HPSD matrix $B$, firstly we rescale the matrix with the largest eigenvalue $\lambda_{\max}$ as $B'=B/(a\lambda_{\max})$ with $a>1$, and find its unitary diagonalization such as $B'=UDU^{\dagger}$ so that we can find a GBS circuit $U$ with thermal input state $\{n_i\}_{i=1}^M$, whose probability matches $\Per(B')$. Then,
\begin{equation}
    \Per(B)=(a\lambda_{\max})^M \Per (B')=(a\lambda_{\max})^M {\cal Z}'p_{\text{th}}.
\end{equation}
If an estimator of $\text{Per}(B)$ lies in the interval $[-C^M,C^M]$, by Hoeffding's inequality~\cite{hoeffding1994probability},
\begin{equation}\label{eq:esper}
    \text{Pr}(|\Per (B)-(a\lambda_{\max})^M {\cal Z}'\mu|\geq (a\lambda_{\max})^M {\cal Z}'\epsilon) \leq 2 \exp \left(-\frac{N\epsilon^2}{2C^{2M}}\right),
\end{equation}
where $\mu$ is the sample mean of $p_\text{th}$. Thus for the number of samples $N=O(\log{\delta^{-1}}/\epsilon^2)$, we can estimate $\Per (B)$ with a success probability $1-\delta$ within an additive-error $\epsilon(a\lambda_{\max}C {\cal Z}'^{1/M})^M$.

Now we use the same method as in the hafnian case by introducing the shifting parameter $\gamma \in [0,1)$, such as
\begin{align}\label{eq:perfor}
    p_\text{th}&=\int d^{2M}\bm{\alpha} \prod^M_{i=1}\frac{2}{\pi(2n_i+1-s)} e^{-\left(\frac{2}{2n_i+1-s}-\gamma\frac{2}{2n_\text{max}+1-s}\right)|\alpha_i|^2}
     \prod_{j=1}^M \frac{8|\beta_j|^2+2(s^2-1)}{(s+1)^3} e^{-\left(\frac{2}{s+1}+\gamma\frac{2}{2n_\text{max}+1-s}\right)|\beta_j|^2} \\
    &=\int d^{2M}\bm{\alpha} \prod^M_{i=1}\frac{2}{\pi}\frac{2n_{\max}+1-s-\gamma(2n_i+1-s)}{(2n_i+1-s)(2n_{\max}+1-s)} 
     e^{-\left(\frac{2}{2n_i+1-s}-\gamma\frac{2}{2n_\text{max}+1-s}\right)|\alpha_i|^2} \nonumber \\
    &\times \prod_{j=1}^M \frac{2n_{\max}+1-s}{2n_{\max}+1-s-\gamma(2n_j+1-s)} 
    \frac{8|\beta_j|^2+2(s^2-1)}{(s+1)^3} e^{-\left(\frac{2}{s+1}+\gamma\frac{2}{2n_\text{max}+1-s}\right)|\beta_j|^2}\\
    &\coloneqq \int d^{2M}\bm{\alpha}\prod^M_{i=1}P_{\text{th},i}(\alpha_i,n_i,\gamma,s)\prod^M_{j=1}f_{\text{th},j}(\beta_j,n_j,\gamma,s).\label{eq:th}
\end{align}
One can compute the upper bound of $|f_{\text{th},j}(\beta_j,n_j,\gamma,s)|$ in three different regions of $\lambda_{\min}=\frac{n_{\min}}{1+n_{\min}} \in [0,1)$, as $\lambda_{\min}=0$, $0<\lambda_{\min}<1/2$, and $1/2 \leq \lambda_{\min}<1$. We set $s=s_{\max}=2n_{\min}+1$, $n_j=\frac{\lambda_j}{1-\lambda_j}$ and note that the extreme points of $f_{\text{th},j}(\beta_j,\lambda_j,\gamma,s)$ are at 0, $\pm \beta^*$ with
\begin{align}
 \beta^*=\sqrt{\frac{\lambda_{\min} (\gamma -2 \lambda_{\min}+1)-\lambda_{\max} ((\gamma -2) \lambda_{\min}+1)}{(\lambda_{\min}-1)^2 (\gamma  (\lambda_{\max}-1)-\lambda_{\max}+\lambda_{\min})}}.
\end{align}

For $\lambda_{\min}=0$, we can find $\gamma$ satisfying $\frac{{\partial f_j(\beta^*,\lambda_{\max},\gamma,s_{\max})}}{{\partial \gamma}}=0$ as
 \begin{equation}
     \gamma^*=\frac{1-2\lambda_{\max}}{2(1-\lambda_{\max})}~~\text{ for~~} 0 \leq \lambda_{\max} < \frac{1}{2}.
\end{equation}
Then an upper bound on the $|f_{\text{th},j}(\beta_j,\lambda_j,\gamma,s)|$ is obtained as
\begin{align}
    &\min_{s,\gamma}\max_{\beta_j}|f_{\text{th},j}(\beta_j,\lambda_j,\gamma,s)| \leq \frac{4\lambda_{\max}^2(1-\lambda_j)}{e(2\lambda_{\max}-\lambda_j)}~~\text{ for~~} 0 \leq \lambda_j <\frac{1}{2}.
\end{align}
 
For $0<\lambda_{\min}<1/2$, we similarly obtain $\gamma$ as
\begin{align}
    &\gamma^*=\frac{\lambda_{\min}+\lambda_{\max}(4\lambda_{\min}-2)+D-\lambda_{\min}\left(3\lambda_{\min}+D\right) }{2\lambda_{\min}(\lambda_{\max}-1)} ~~\text{for~~~}  \lambda_j \leq \frac{\lambda_{\min}^2+\lambda_{\min}-1}{3\lambda_{\min}-2},
\end{align}
where $D=\sqrt{4\lambda_{\max}^2-8\lambda_{\max} \lambda_{\min}+5\lambda_{\min}^2}$. Then an upper bound on the $|f_{\text{th},j}(\beta_j,\lambda_j,\gamma,s)|$ is given by
\begin{align}
\min_{s,\gamma}\max_{\beta_j}|f_{\text{th},j}(\beta_j,\lambda_j,\gamma,s)|&\leq \frac{4 (1-\lambda_j) \lambda_{\min}^2 e^{\frac{\lambda_{\min}-D}{2 \lambda_{\max}-2 \lambda_{\min}}} (\lambda_{\max}-\lambda_{\min})^2}{\left(D-2 \lambda_{\max}+\lambda_{\min}\right) \left(\lambda_{\min} \left(D-4 \lambda_{\max}+3 \lambda_{\min}\right)-\lambda_j \left(D-2 \lambda_{\max}+\lambda_{\min}\right)\right)}\\
    &\text{for~~~}  \lambda_j \leq \frac{\lambda_{\min}^2+\lambda_{\min}-1}{3\lambda_{\min}-2}.
\end{align}

We will consider the case of $1/2 \leq \lambda_{\min} <1$ later, in which we achieve a multiplicative-error estimation scheme. To cover the full range of $\lambda_j$, we consider the reverse shifting such that
\begin{align}
    p_\text{th}&=\int d^{2M}\bm{\alpha} \prod^M_{i=1}\frac{2}{\pi(2n_i+1-s)} e^{-\left(\frac{2}{2n_i+1-s}+\gamma'\frac{2}{s+1}\right)|\alpha_i|^2}
     \prod_{j=1}^M \frac{8|\beta_j|^2+2(s^2-1)}{(s+1)^3} e^{-(1-\gamma')\frac{2}{s+1}|\beta_j|^2} \\
    &=\int d^{2M}\bm{\alpha} \prod^M_{i=1}\frac{2}{\pi}\frac{s+1+\gamma'(2n_i+1-s)}{(2n_i+1-s)(s+1)} 
     e^{-\left(\frac{2}{2n_i+1-s}+\gamma'\frac{2}{s+1}\right)|\alpha_i|^2} \nonumber \\
    &\times \prod_{j=1}^M \frac{s+1}{s+1+\gamma'(2n_j+1-s)}
    \frac{8|\beta_j|^2+2(s^2-1)}{(s+1)^3} e^{-(1-\gamma')\frac{2}{s+1}|\beta_j|^2}\\
    &\coloneqq \int d^{2M}\bm{\alpha}\prod^M_{i=1}P'_{\text{th},i}(\alpha_i,n_i,\gamma',s)\prod^M_{j=1}f'_{\text{th},j}(\beta_j,n_j,\gamma',s).
\end{align}
where $\gamma' \in (0,1]$. When we put $s=2n_{\min}+1$, the extreme points of $f'_{\text{th},j}(\beta_j,\lambda_j,\gamma',s)$ are at 0, $\pm \beta^{'*}$ with
\begin{equation}
    \beta^{'*}=\sqrt{\frac{-\gamma'  \lambda_{\min}+2 \lambda_{\min}-1}{\gamma'  \lambda_{\min}^2-2 \gamma'  \lambda_{\min}+\gamma' -\lambda_{\min}^2+2 \lambda_{\min}-1}}.
\end{equation}
For $\lambda_{\min}=0$, the condition $\frac{{\partial f_j(\beta^{'*},\lambda_{\max},\gamma',s_{\max})}}{{\partial \gamma'}}=0$ yields
\begin{equation}
    \gamma'^*=\frac{2\lambda_{\max}-1}{2\lambda_{\max}},~~\text{ for~~}  \frac{1}{2} \leq \lambda_{\max} <1.
\end{equation}
Then an upper bound of $|f'_{\text{th},j}(\beta_j,\lambda_j,\gamma',s)|$ is given by
\begin{align}
&\min_{s,\gamma'}\max_{\beta_j}|f'_{\text{th},j}(\beta_j,\lambda_j,\gamma',s)|\leq \frac{4\lambda_{\max}^2(1-\lambda_j)}{e(2\lambda_{\max}-\lambda_j)},~~\text{ for~~}  \frac{1}{2} \leq \lambda_j <1,
\end{align}
which is consistent with the bound for $0<\lambda_j<\frac{1}{2}$.

Similarly, when $0<\lambda_{\min}<1/2$,
\begin{equation}
\gamma'^*=\frac{\lambda_{\min}+\lambda_{\max}(4\lambda_{\min}-2)+D-\lambda_{\min}\left(3\lambda_{\min}+D \right)}{2\lambda_{\min}(\lambda_{\max}-\lambda_{\min})},~~\text{for~~}  \frac{\lambda_{\min}^2+\lambda_{\min}-1}{3\lambda_{\min}-2}\leq \lambda_j <1
\end{equation}
and the corresponding bound is
\begin{align}
\min_{s,\gamma'}\max_{\beta_j}|f'_{\text{th},j}(\beta_j,\lambda_j,\gamma',s)|&\leq \frac{4 (1-\lambda_j) \lambda_{\min}^2 e^{\frac{\lambda_{\min}-D}{2 \lambda_{\max}-2 \lambda_{\min}}} (\lambda_{\max}-\lambda_{\min})^2}{\left(D-2 \lambda_{\max}+\lambda_{\min}\right) \left(\lambda_{\min} \left(D-4 \lambda_{\max}+3 \lambda_{\min}\right)-\lambda_j \left(D-2 \lambda_{\max}+\lambda_{\min}\right)\right)},\\
    &\text{for~~~}  \frac{\lambda_{\min}^2+\lambda_{\min}-1}{3\lambda_{\min}-2}\leq \lambda_j <1.
\end{align}

Meanwhile,
\begin{equation}
    {\cal Z'}=\prod^M_{i=1}(1+n_i)=\frac{1}{\prod^M_{i=1}(1-\lambda_i)}.
\end{equation} 
Consequently, for the number of samples $N=O(\log{\delta^{-1}}/\epsilon^2)$ and success probability $1-\delta$, we can estimate $\Per(B)$ when the minimum eigenvalue $\lambda_{\min}=0$ within an additive-error as
\begin{align}\label{eq:per1}
    \epsilon \prod^M_{i=1}\frac{4a\lambda_{\max}\lambda'^{2}_{\max}}{e(2\lambda'_{\max}-\lambda'_i)}=\epsilon \prod_{i=1}^M \frac{4\lambda^{2}_{\max}}{e(2\lambda_{\max}-\lambda_i)},
\end{align}
where $\lambda_i'=\frac{\lambda_i}{a\lambda_{\max}}$. Let us first compare our result with the existing result~\cite{chakhmakhchyan2017quantum}, where $\gamma=0$ and $s=1$. In the latter case, the upper bound on the estimator is such as $\max |f_j|\leq e^{-1}$, so thus corresponding additive-error is given by
\begin{equation}\label{eq:pra}
    \epsilon \prod^M_{i=1}\frac{1}{e(1-\lambda_i)},
\end{equation}
where we assume $\lambda_i \in [0,1)$.
Note that Eq.~(\ref{eq:per1}) is smaller than Eq.~(\ref{eq:pra}) when $\lambda_{\max} \in (0,1/2)$, so we have a better precision.

Next, we compare our algorithm's precision with Gurvits' randomized algorithm for the permanent of a general complex matrix $A$ giving the additive-error as $\epsilon ||A||^M=\epsilon\lambda_{\max}^M$ with samples $N=O(M^2/\epsilon^2)$. Thus from Eq.~(\ref{eq:per1}), the necessary and sufficient condition for beating Gurvits' precision is written as
\begin{equation}
    \prod_{i=1}^M \frac{4\lambda^{2}_{\max}}{e(2\lambda_{\max}-\lambda_i)} < \lambda_{\max}^M.
\end{equation}
To get lower and upper bounds, we put $\lambda_i=0$ and $\lambda_i=\lambda_{\max}$ for all $i$, respectively, such as
\begin{equation}
   \epsilon(0.736\lambda_{\max})^M \simeq \epsilon\prod_{i=1}^M \frac{2\lambda_{\max}}{e} \leq  \epsilon\prod_{i=1}^M \frac{4\lambda^{2}_{\max}}{e(2\lambda_{\max}-\lambda_i)} \leq \epsilon\prod_{i=1}^M \frac{4\lambda_{\max}}{e} \simeq \epsilon(1.472\lambda_{\max})^M.
\end{equation}
Also for $0<\lambda_{\min}<1/2$, the additive-error for $\Per(B)$ is given by
\begin{align}\label{eq:per2}
    &\epsilon \prod^M_{i=1}\frac{4 a\lambda_{\max} \lambda'^2_{\min} e^{\frac{\lambda'_{\min}-D'}{2 \lambda'_{\max}-2 \lambda'_{\min}}} (\lambda'_{\max}-\lambda'_{\min})^2}{\left(D'-2 \lambda'_{\max}+\lambda'_{\min}\right) \left(\lambda'_{\min} \left(D'-4 \lambda'_{\max}+3 \lambda'_{\min}\right)-\lambda'_i \left(D'-2 \lambda'_{\max}+\lambda'_{\min}\right)\right)}\\
    &=\epsilon \prod_{i=1}^M \frac{4 \lambda_{\min}^2 e^{\frac{\lambda_{\min}-D}{2 \lambda_{\max}-2 \lambda_{\min}}} (\lambda_{\max}-\lambda_{\min})^2}{\left(D-2 \lambda_{\max}+\lambda_{\min}\right) \left(\lambda_{\min} \left(D-4 \lambda_{\max}+3 \lambda_{\min}\right)-\lambda_i \left(D-2 \lambda_{\max}+\lambda_{\min}\right)\right)} \coloneqq \epsilon \prod^M_{i=1}H^B_i(\lambda_i),
\end{align}
where $D'=\sqrt{4\lambda'^2_{\max}-8\lambda'_{\max} \lambda'_{\min}+5\lambda'^2_{\min}}$.
Similarly, the necessary and sufficient condition for beating Gurvits' is given by
\begin{equation}
    \prod^M_{i=1}H^B_i(\lambda_i)<\lambda_{\max}^M.
\end{equation}

\subsection{Torontonian}\label{app:addtor}

The Torontonian of a $2M \times 2M$ complex matrix $A'$ is defined as~\cite{quesada2018gaussian}
\begin{equation}
    \text{Tor}(A')=\sum_{Z \in P([M])}(-1)^{|Z|}\frac{1}{\sqrt{\det(\mathbb{I}-A'_{(Z)})}},
\end{equation}
where $P([M])$ is the power set of $[M]:=\{1,2,...,M \}$ and the matrix $A'$ has block structure such that
\begin{equation}\label{eq:aprime}
    A'=\begin{pmatrix} B^T & R^* \\ R & B \end{pmatrix},
\end{equation}
where $B$ is HPSD and $R$ is symmetric. Let us first consider a special case $B=0$, which corresponds to a GBS with pure squeezed input state $\{r_i\}_{i=1}^M$.
The probability of all threshold detectors ``click" in a GBS circuit is related to the Torontonian as
\begin{equation}
    p_{\text{on}|\text{sq}}=\frac{1}{\sqrt{|V_Q|}}\text{Tor}(\mathbb{I}_{2M}-V^{-1}_Q)=\frac{1}{\cal Z}\text{Tor}\begin{pmatrix} 0 & R^* \\ R & 0\end{pmatrix},
\end{equation}
where $V_Q=V+\mathbb{I}/2$ with the covariance matrix $V$, $\Pi_0=\ket{0}\bra{0}$, ${\cal Z}=\prod_i^{M}\cosh{r_i}$, $R=UDU^T$, and $D=\oplus^M_i \tanh{r_i}$. Meanwhile, the probability $p_{\text{on}|\text{sq}}$ can be written in terms of $s$-PQDs as
\begin{align}
     p_{\text{on}|\text{sq}}&=
     \int d^{2M}\bm{\alpha}\prod^M_{i=1}P_{\text{sq},i}(\alpha_i,r_i,\gamma,s)\prod^M_{j=1}f_{\text{on}|\text{sq},j}(\beta_j,r_j,\gamma,s),
\end{align}
where $P_{\text{sq},i}(\alpha_i,r_i,\gamma,s)$ is given in Eq.~(\ref{eq:sq}) and $f_{\text{on}|\text{sq},j}(\beta_j,r_j,\gamma,s)$ is defined as 
\begin{align}
    f_{\text{on}|\text{sq},j}(\beta_j,r_j,\gamma,s)=\sqrt{\frac{e^{2r_{\max}}-s}{e^{2r_{\max}}-s-\gamma(e^{2r_j}-s)}}\sqrt{\frac{e^{2r_{\max}}-s}{e^{2r_{\max}}-s-\gamma(e^{-2r_j}-s)}}\left(1-\frac{2}{s+1}e^{-\frac{2}{s+1}|\beta_j|^2}\right) e^{-\gamma \frac{2}{e^{2r_{\max}}-s}|\beta_j|^2}.
\end{align}
We set $s=s_{\max}=e^{-2r_{\max}}$ and $r_j=\tanh^{-1}\lambda_j$. The extreme points are at $0, \pm \beta^*$ with
\begin{align}
    \beta^*=\sqrt{\frac{1}{\lambda_{\max}+1}\log \left(\frac{(\lambda_{\max}+1) (\gamma  (\lambda_{\max}-1)-2 \lambda_{\max})}{\gamma  (\lambda_{\max}-1)}\right)}.
\end{align}
After we choose $\gamma^*=\frac{1}{2}(1-\lambda_{\max})$, the corresponding upper bound is given by
\begin{align}
     &\min_{s,\gamma}\max_{\beta_j} \left| 
    f_{\text{on},j}(\beta_j,\lambda_j,\gamma,s) \right| \\
    &\leq \frac{16\sqrt{1-\lambda_i^2}\lambda_{\max}^2}{(1+\lambda_{\max}^3)\sqrt{\{\lambda_i+(\lambda_i+3)\lambda_{\max}-\lambda_{\max}^2\}\{-\lambda_i-(\lambda_i-3)\lambda_{\max}-\lambda_{\max}^2 \}}}\left[\frac{(1+\lambda_{\max})^3}{(1-\lambda_{\max})^2} \right]^{-\frac{(1-\lambda_{\max})^2}{4\lambda_{\max}}},
\end{align}
Therefore, with success probability $1-\delta$, we can estimate the Torontonian of a matrix $\begin{pmatrix}0 & R^* \\ R & 0 \end{pmatrix}$ using the number of samples $O(\log{\delta^{-1}}/\epsilon^2)$ within an additive-error 
\begin{align}\label{eq:addtorr}
    \epsilon \prod_{i=1}^M \frac{16\lambda_{\max}^2}{(1+\lambda_{\max}^3)\sqrt{\{\lambda_i+(\lambda_i+3)\lambda_{\max}-\lambda_{\max}^2\}\{-\lambda_i-(\lambda_i-3)\lambda_{\max}-\lambda_{\max}^2 \}}}\left[\frac{(1+\lambda_{\max})^3}{(1-\lambda_{\max})^2} \right]^{-\frac{(1-\lambda_{\max})^2}{4\lambda_{\max}}}\coloneqq \epsilon \prod_{i=1}^M T_i(\lambda_i).
\end{align}

Next, we consider a special case where $R=0$ corresponds to a thermal input state $\{n_i\}_{i=1}^M$. In that case, 
\begin{equation}\label{eq:tor_th}
    p_{\text{on}|\text{th}}=\frac{1}{{\cal Z}'}\text{Tor}\begin{pmatrix}B^T & 0 \\ 0 & B \end{pmatrix},
\end{equation}
where ${\cal Z'}=\prod_{i=1}^M (1+n_i)$, $B=UDU^{\dagger}$, and $D=\text{diag}\{ \frac{n_1}{n_1+1},...,\frac{n_M}{n_M+1}\}$. Meanwhile, the probability $p_{\text{on}|\text{th}}$ can also be written as 
\begin{align}
   p_{\text{on}|\text{th}}=\int d^{2M}\bm{\alpha}\prod^M_{i=1}P_{\text{th},i}(\alpha_i,n_i,\gamma,s)\prod^M_{j=1}f_{\text{on}|\text{th},j}(\beta_j,n_j,\gamma,s),
\end{align}
where $P_{\text{th},i}(\alpha_i,n_i,\gamma,s)$ is given in Eq.~(\ref{eq:th}) and $f_{\text{on}|\text{th},j}(\beta_j,n_j,\gamma,s)$ is defined as
\begin{equation}\label{eq:fonth}
 f_{\text{on}|\text{th},j}(\beta_j,n_j,\gamma,s)=\frac{2n_{\max}+1-s}{2n_{\max}+1-s-\gamma(2n_j+1-s)} 
    \left(1-\frac{2}{s+1}e^{-\frac{2}{s+1}|\beta_j|^2}\right) e^{-\gamma\frac{2}{2n_\text{max}+1-s}|\beta_j|^2}.
\end{equation}
We set $s=s_{\max}=2n_{\min}+1$ and $n_j=\frac{\lambda_j}{1-\lambda_j}$. The extreme points are at $0,\pm \beta^*$ with
\begin{equation}
    \beta^*=\sqrt{\frac{1}{\lambda_{\min}-1}\log \left(\frac{\gamma -\gamma  \lambda_{\max}}{(\lambda_{\min}-1) (\gamma  \lambda_{\max}-\gamma -\lambda_{\max}+\lambda_{\min})}\right)}.
\end{equation}
Choosing $\gamma^*=\frac{1}{2}(1-\lambda_{\max})$, an upper bound on $|f_{\text{on}|\text{th},j}(\beta_j,\lambda_j,\gamma,s)|$ is given by
\begin{align}
     \min_{s,\gamma}\max_{\beta_j} \left| 
    f_{\text{on}|\text{th},j}(\beta_j,\lambda_j,\gamma,s) \right| &\leq \frac{4(1-\lambda_j)(\lambda_{\max}-\lambda_{\min})^2\left(\frac{(1-\lambda_{\max})^2}{(1-\lambda_{\min})(1+\lambda_{\max}^2-2\lambda_{\min})} \right)^{\frac{1+\lambda_{\max}^2-2\lambda_{\min}}{2\lambda_{\max}-2\lambda_{\min}}}(\lambda_{\min}-1)}{(1-\lambda_{\max})^2\{ \lambda_j(1+\lambda_{\max}^2-2\lambda_{\min})+\lambda_{\min}-\lambda_{\max}(2+(\lambda_{\max}-2)\lambda_{\min}))\}}.
\end{align}
Therefore, we can estimate Torontonian of a matrix $\begin{pmatrix}B^T & 0 \\ 0 & B \end{pmatrix}$ with success probability $1-\delta$ using number of samples $O(\log{\delta^{-1}}/\epsilon^2)$ within an additive-error 
\begin{align}\label{eq:addtorb}
    \epsilon \prod_{i=1}^M \frac{4(\lambda_{\max}-\lambda_{\min})^2\left[\frac{(1-\lambda_{\max})^2}{(1-\lambda_{\min})(1+\lambda_{\max}^2-2\lambda_{\min})} \right]^{\frac{1+\lambda_{\max}^2-2\lambda_{\min}}{2\lambda_{\max}-2\lambda_{\min}}}(\lambda_{\min}-1)}{(1-\lambda_{\max})^2\{ \lambda_i(1+\lambda_{\max}^2-2\lambda_{\min})+\lambda_{\min}-\lambda_{\max}(2+(\lambda_{\max}-2)\lambda_{\min}))\}} \coloneqq \epsilon \prod_{i=1}^M  T_i^B(\lambda_i).
\end{align}

Finally, we consider a matrix $A'$ defined by Eq.~(\ref{eq:aprime}), satisfying Eqs.~(\ref{eq:acondition1}) and (\ref{eq:acondition2}). Then the Torontonian of matrix $A'$ is related with a squeezed thermal input state $\{r_i,n\}_{i=1}^M$ such that
\begin{equation}
   p_{\text{on}|\text{st}} =\frac{\text{Tor}(A')}{\sqrt{|V_Q|}}.
\end{equation}
Meanwhile, $p_{\text{on}|\text{st}}$ is also can be written as
\begin{align}
   &p_{\text{on}|\text{st}}=\int d^{2M}\bm{\alpha} \prod^{M}_{i=1}\frac{2}{\pi \sqrt{(a_+(r_i,n)-s)(a_-(r_i,n)-s)}}e^{-\frac{2\alpha_{ix}^2}{a_+(r_i,n)-s}-\frac{2\alpha_{iy}^2}{a_-(r_i,n)-s}} \prod_{j=1}^{M} \left( 1-\frac{2}{s+1}e^{-\frac{2}{s+1}|\beta_j|^2}\right) \\
   &=\int d^{2M}\bm{\alpha} \prod^{M}_{i=1}\frac{2}{\pi \sqrt{(a_+(r_i,n)-s)(a_-(r_i,n)-s)}}e^{-\left(\frac{2}{a_+(r_i,n)-s}-\gamma \frac{2}{a_{\max}-s} \right)\alpha_{ix}^2-\left(\frac{2}{a_-(r_i,n)-s}-\gamma \frac{2}{a_{\max}-s} \right)\alpha_{iy}^2} \nonumber\\
    &\times \prod_{j=1}^{M} \left( 1-\frac{2}{s+1}e^{-\frac{2}{s+1}|\beta_j|^2}\right) e^{-\gamma \frac{2}{a_{\max}-s}|\beta_j|^2} \\
    &=\int d^{2M}\bm{\alpha} \prod^{M}_{i=1}\frac{2}{\pi}\sqrt{\frac{a_{\max}-s-\gamma(a_+(r_i,n)-s)}{(a_+(r_i,n)-s)(a_{\max}-s)}}\sqrt{\frac{a_{\max}-s-\gamma(a_-(r_i,n)-s)}{(a_-(r_i,n)-s)(a_{\max}-s)}}e^{-\left(\frac{2}{a_+(r_i,n)-s}-\gamma \frac{2}{a_{\max}-s} \right)\alpha_{ix}^2-\left(\frac{2}{a_-(r_i,n)-s}-\gamma \frac{2}{a_{\max}-s} \right)\alpha_{iy}^2}\nonumber \\
    &\times \prod_{j=1}^{M}\sqrt{\frac{a_{\max}-s}{a_{\max}-s-\gamma(a_+(r_j,n)-s)}}\sqrt{\frac{a_{\max}-s}{a_{\max}-s-\gamma(a_-(r_j,n)-s)}} \left( 1-\frac{2}{s+1}e^{-\frac{2}{s+1}|\beta_j|^2}\right) e^{-\gamma \frac{2}{a_{\max}-s}|\beta_j|^2}\\
    &\coloneqq \int d^{2M}\bm{\alpha}\prod^M_{i=1}P_{\text{st},i}(\alpha_i,r_i,n,\gamma,s)\prod^M_{j=1}f_{\text{on}|\text{st},j}(\beta_j,r_j,n,\gamma,s).
\end{align}
where $a_{\pm}(r_i,n)=(2n+1)e^{\pm 2r_i}$, $a_{\max}=(2n+1)e^{2r_{\max}}$, and $a_{\min}=(2n+1)e^{-2r_{\max}}$. We set $s=s_{\max}=a_{\min}$ and note the extreme points are $0, \pm \beta^*$ with
\begin{equation}
    \beta^*=\sqrt{\frac{1}{2} \left((2 n+1) e^{-2 r_{\max}}+1\right) \log \left(\frac{2 e^{2 r_{\max}} \left((\gamma -1) (2 n+1)+(2 n+1) e^{4 r_{\max}}+\gamma  e^{2 r_{\max}}\right)}{\gamma  \left(2 n+e^{2 r_{\max}}+1\right)^2}\right)}
\end{equation}

By choosing $\gamma^*=\frac{e^{-\tanh{r_{\max}}}}{n+1}$, 
an upper bound on $|f_{\text{on}|\text{st},j}(\beta_j,r_j,n,\gamma,s)|$ is given by
\begin{align}
    &|f_{\text{on}|\text{st},j}(\beta_j,r_j,n,\gamma,s)|\leq f_{\text{on}|\text{st},j}(\beta^*_j,r_j,n,\gamma^*,a_{\min})\\
    &=\frac{(n+1)^2 (2 n+1) \left(e^{4 r_{\max}}-1\right)^2 e^{r_j+2 r_{\max}+2 \tanh (r_{\max})}}{\left(2 n+e^{2 r_{\max}}+1\right)^2} \sqrt{\frac{1}{(n+1) \left(-e^{\tanh (r_{\max})}\right)+(n+1) e^{4 r_{\max}+\tanh (r_{\max})}-e^{2 (r_j+r_{\max})}+1}} \\
    &\times \sqrt{\frac{1}{(n+1) \left(-e^{2 r_j+\tanh (r_{\max})}\right)+(n+1) e^{2 r_j+4 r_{\max}+\tanh (r_{\max})}+e^{2 r_j}-e^{2 r_{\max}}}} 2^{F} \\
    &\times \left(\frac{e^{2 r_{\max}} \left(-\left(2 n^2+3 n+1\right) e^{\tanh (r_{\max})}+\left(2 n^2+3 n+1\right) e^{4 r_{\max}+\tanh (r_{\max})}+2 n+e^{2 r_{\max}}+1\right)}{\left(2 n+e^{2 r_{\max}}+1\right)^2}\right)^{F}, \\
    &F=-\frac{e^{-\tanh (r_{\max})} (n \tanh (r_{\max})+(n+1) \coth (r_{\max})-2 n-1)}{2 (n+1) (2 n+1)}.
\end{align}
Therefore, we can estimate the Torontonian of a matrix $\begin{pmatrix} C^T & R^* \\ R & C \end{pmatrix}$ with a success probability $1-\delta$ for number of samples $O(\log{\delta^{-1}}/\epsilon^2)$ within an additive-error as
\begin{equation}\label{eq:addtora}
    \epsilon \prod_{i=1}^M \sqrt{\frac{1}{2}+n(n+1)+(n+\frac{1}{2})\cosh{2r_i}}~f_{\text{on}|\text{st},i}(\beta^*,r_i,n,\gamma^*,a_{\min})\coloneqq \epsilon \prod_{i=1}^M T_i^A(n,r_i).
\end{equation}

\section{Estimation of matrix functions within multiplicative-errors}

The previous section investigates the efficient estimation schemes for various matrix functions within additive-errors. This section proposes a much stronger scheme, such as estimations within multiplicative-errors. We show that for highly classical input states, the $s$-PQDs representations of the outcome probabilities corresponding to those matrix functions can be written as integrals of log-concave functions. Then we can use an FPRAS to approximate the matrix functions.

\subsection{Permanent of Hermitian positive definite matrices}\label{app:mulper}
Recently, it has been proven that the multiplicative-error estimation of an HPSD matrix is NP-hard~\cite{meiburg2021inapproximability}. The case of a Hermitian positive definite matrix, i.e., $\lambda_{\min}>0$, is still unknown yet, but introduced an FPRAS for $1\leq \lambda_i \leq 2$~\cite{barvinok2021remark}. In this section, we reproduce the same result by showing the log-concavity for $\lambda_{\max}/\lambda_{\min}\leq 2$. By a slight generalization of the result in Ref.~\cite{barvinok2021remark}, we give a following lemma:
\begin{lemma}\label{lem2}
Let $a,b,c \geq 0$ and $q:\mathbb{R}^M \rightarrow \mathbb{R}_+$ is a positive semidefinite quadratic form. Then the function $(a+bq(x))e^{-cq(x)}$ is log-concave when $ca \geq b$.
\end{lemma}
\begin{proof}
It is enough to show the function
\begin{equation}
    h(x)=\log (a+bq(x))-cq(x)
\end{equation}
is concave when $ca \geq b > 0$. This leads to check that $h$ onto any affine line $x(\tau)=\alpha\tau+\beta$ with $\alpha,\beta \in \mathbb{R}^M$ is concave. Via the affine substitution $\tau \coloneqq (\tau-\beta)/\alpha$, what we need to check is $g''(\tau) \leq 0$ for all $\tau \in \mathbb{R}$, where
\begin{equation}
    g(\tau)=\log(a+b(\tau^2+\gamma^2))-c(\tau^2+\gamma^2).
\end{equation}
Then by straightforward calculation,
\begin{align}
    g''(\tau)&=-2c-\frac{4b^2\tau^2}{(a+b(\gamma^2+\tau^2))^2}+\frac{2b}{a+b(\gamma^2+\tau^2)} \nonumber \\
    &=\frac{2b(a+b(\gamma^2+\tau^2))-2c(a+b(\gamma^2+\tau^2))^2-4b^2\tau^2}{(a+b(\gamma^2+\tau^2))^2} \nonumber \\
    &=\frac{-2a(ca-b)-2b(\gamma^2+\tau^2)(ca-b)-2abc(\gamma^2+\tau^2)-4b^2\tau^2-2b^2(\gamma^2+\tau^2)^2}{(a+b(\gamma^2+\tau^2))^2} \nonumber \\
    &\leq 0~~\text{ for } ca \geq b. \nonumber
\end{align}
\end{proof}
From Eq.~(\ref{eq:perfor}),
\begin{equation}
    f_{\text{th},j}(\beta_j,n_i,\gamma,s)\propto
    \frac{8|\beta_j|^2+2(s^2-1)}{(s+1)^3} e^{-\left(\frac{2}{s+1}+\gamma\frac{2}{2n_\text{max}+1-s}\right)|\beta_j|^2}.
\end{equation}
Therefore, the condition for log-concavity of $ f_{\text{th},j}(\beta_j,n_i,\gamma,s)$ is written as
\begin{equation}
    \left(\frac{2}{s+1}+\gamma\frac{2}{2n_\text{max}+1-s}\right)2(s^2-1)\geq 8.
\end{equation}
After putting $\gamma=1$ and $s=s_{\max}=2n_{\min}+1$,
\begin{equation}
    \frac{n_{\max}}{n_{\max}+2}\leq n_{\min}.
\end{equation}
Substituting $n_i=\frac{\lambda_i}{1-\lambda_i}$,
\begin{equation}
    \frac{\lambda_{\max}}{\lambda_{\min}}\leq 2,
\end{equation}
as desired.

\subsection{Hafnian (Proof of Theorem~\ref{th:FPRAS})}\label{app:FPRAShaf}

When the input is a squeezed thermal state $\{r_i,n\}_{i=1}^M$, we can find out the condition for multiplicative estimation of the hafnian of a particular matrix. Note that the estimated function in the case of all single-photon outcomes is given by 

\begin{equation}
    f_{\text{st},j}(\beta_j,r_j,n,\gamma,s) \propto \frac{8|\beta_j|^2+2(s^2-1)}{(s+1)^3} e^{-\left(\frac{2}{s+1}+\gamma\frac{2}{a_{\max}-s}\right)|\beta_j|^2},
\end{equation}
where $a_{\max}=(2n+1)e^{2r_{\max}}$. Then from Lemma.~\ref{lem2}, the condition for log-concavity of $f_{\text{st},j}(\beta_j,r_j,n,\gamma,s)$ is 
\begin{equation}\label{eq:hafcondition}
    n \geq \frac{1}{4} \left(6 \sinh (2 r_{\max})+\sqrt{18 \cosh (4 r_{\max})-14}-2\right),
\end{equation}
where we put $\gamma=1$ and $s=s_{\max}=(2n+1)e^{-2r_{\max}}$. Thus for a matrix $A=\begin{pmatrix} R & B \\ B^T & R^*\end{pmatrix}$ satisfying conditions Eqs.~(\ref{eq:acondition1}) and (\ref{eq:acondition2}), $\text{Haf}(A)$ can be estimated within a multiplicative-error efficiently when Eq.~(\ref{eq:hafcondition}) holds. 

\subsection{Torontonian}\label{app:FPRAStor}

We can apply the same method to the Torontonian. Here, we use the following lemma:
\begin{lemma}\label{lem1}
Let $a,b,c \geq 0$ and $q:\mathbb{R}^M \rightarrow \mathbb{R}_+$ is a positive semidefinite quadratic form. Then the function $(a-be^{-bq(x)})e^{-cq(x)}$ is log-concave when $a\geq \frac{b^2+2bc}{c}$.
\end{lemma}
\begin{proof}
By the same argument in Lemma.~\ref{lem2}, what we need to check is $g''(\tau) \leq 0$ for all $\tau \in \mathbb{R}$, where
\begin{equation}
    g(\tau)=\log(a-be^{-b(\tau^2+\gamma^2)})-c(\tau^2+\gamma^2).
\end{equation}
By a straightforward calculation,
\begin{align}
    g''(\tau)&=-2c+\frac{-2b^3-2ab^2(2b\tau^2-1)e^{b(\gamma^2+\tau^2)}}{\left(b-ae^{b(\gamma^2+\tau^2)}\right)^2} \\
    &=\frac{-2c\left(b^2-2abe^{b(\gamma^2+\tau^2)}+a^2e^{2b(\gamma^2+\tau^2)}\right)-2b^3-4ab^3\tau^2e^{b(\gamma^2+\tau^2)}+2ab^2e^{b(\gamma^2+\tau^2)}}{\left(b-ae^{b(\gamma^2+\tau^2)}\right)^2}\\
    &=\frac{-2cb^2e^{-b(\gamma^2+\tau^2)}+4abc-2a^2ce^{b(\gamma^2+\tau^2)}-2b^3e^{-b(\gamma^2+\tau^2)}-4ab^3\tau^2+2ab^2}{e^{-b(\gamma^2+\tau^2)}\left(b-ae^{b(\gamma^2+\tau^2)}\right)^2}\\
    &\leq \frac{-2cb^2e^{-b(\gamma^2+\tau^2)}+4abc-2a^2c-2b^3e^{-b(\gamma^2+\tau^2)}-4ab^3\tau^2+2ab^2}{e^{-b(\gamma^2+\tau^2)}\left(b-ae^{b(\gamma^2+\tau^2)}\right)^2}\\
    &\leq 0~~\text{ for } a\geq \frac{b^2+2bc}{c}. \nonumber
\end{align}
\end{proof}

First we consider thermal input state $\{n_i\}_{i=1}^M$. From  Eq.~(\ref{eq:fonth}),
\begin{equation}
    f_{\text{on}|\text{th},j}(\beta_j,n_j,\gamma,s) \propto \left( 1-\frac{2}{s+1}e^{-\frac{2}{s+1}|\beta_j|^2}\right) e^{-\gamma\frac{2}{2n_\text{max}+1-s}|\beta_j|^2}.
\end{equation}
By Lemma.~\ref{lem1}, the function $(a-be^{-b|\beta|^2})e^{-c|\beta|^2}$ is log-concave when $a\geq \frac{b^2+2bc}{c}$. Consequently, the condition for log-concavity of $f_{\text{on}|\text{th},j}(\beta_j,n_j,\gamma,s)$ is written as
\begin{equation}
    \frac{2(3+2n_{\max}+s)}{(1+s)^2}\leq 1.
\end{equation}
When $\gamma=1$ and $s=s_{\max}=2n_{\min}+1$, this condition yields
\begin{equation}\label{eq:torcmulti}
    \frac{1}{2}\leq \lambda_{\min}\leq \lambda_{\max}\leq \frac{-\lambda_{\min}^2+3\lambda_{\min}-1}{\lambda_{\min}},
\end{equation}
where $\lambda_j=\frac{n_j}{1+n_j}$. Thus for an HPSD matrix $B$, $\text{Tor}\begin{pmatrix} B^T & 0 \\ 0 & B \end{pmatrix}$ can be estimated within a multiplicative-error when eigenvalues $\{\lambda_i\}$ of $C$ satisfy the condition Eq.~(\ref{eq:torcmulti}). Note that this condition is more stringent than the permanent case in which there is no restriction for $\lambda_{\max}$ when $\lambda_{\min}\geq \frac{1}{2}$.

Next, when the input is a squeezed thermal state $\{r_i,n\}_{i=1}^M$, $f_{\text{on}|\text{st},j}(\beta_j,r_j,n,\gamma,s)$ is given by
\begin{equation}
    f_{\text{on}|\text{st},j}(\beta_j,r_j,n,\gamma,s) \propto \left( 1-\frac{2}{s+1}e^{-\frac{2}{s+1}|\beta_j|^2}\right) e^{-\gamma\frac{2}{a_{\max}-s}|\beta_j|^2},
\end{equation}
where $a_{\max}=(2n+1)e^{2r_{\max}}$. If we set $\gamma=1$ and $s=s_{\max}=(2n+1)e^{-2r_{\max}}$, the condition for log-concavity of $f_{\text{on}|\text{st},j}(\beta_j,r_j,n,\gamma,s)$ is given by using Lemma.~\ref{lem1}
\begin{equation}\label{eq:tormul}
    n\geq \frac{1}{2}\left(e^{2r_{\max}}\sqrt{e^{8r_{\max}}+3}+e^{6r_{\max}}-1\right).
\end{equation}
Thus for a matrix $A'=\begin{pmatrix} B^T & R^* \\R & B\end{pmatrix}$ satisfying Eqs.~(\ref{eq:acondition1}) and (\ref{eq:acondition2}), we can estimate the Torontonian with multiplicative-error when the above condition is satisfied.

\section{Lower and upper bounds on matrix functions}\label{app:bounds}
So far, we have suggested polynomial-time randomized algorithms for estimating various matrix functions. Here, we provide lower and upper bounds on the matrix functions by using $s$-PQD and appropriately choosing $s$, which is summarized in the following Table~\ref{table:upperlower}.

\begin{table*}[!ht]
\setlength{\tabcolsep}{0.5em}
{\renewcommand{\arraystretch}{2}
\begin{ruledtabular}
\begin{tabularx}{\textwidth}{P{1.5cm}||P{8.5cm}|P{7.5cm}}
     & Upper bound &  Lower bound   \\
    \hline
    \hline
 $|\text{Haf}(R)|^2$ & ?  & ? \\
 \hline
 $\text{Per}(B)$  & $\prod_{i}^M G_i(\lambda_i)~(*)$~Eq.~(\ref{eq:upper})  & $\prod_i^M\frac{\lambda_{\min}^2}{\lambda_i}$~\cite{chakhmakhchyan2017quantum}    \\
 \hline
 Haf($A$) &$\prod_{i}^M G^H_i(r_i,n)~(*)$~Eq.~(\ref{eq:uphaf}) & $\prod_{i}^M L^H_i(r_i,n)~(*)$~Eq.~(\ref{eq:lowhaf})    \\
 \hline
 Tor($R'$) & ? & ?  \\
 \hline
 Tor($B'$) & $\prod_{i}^M \frac{\lambda_{\max}^2}{\lambda_i(1-\lambda_{\max})}~(*)$ & $\prod_{i}^M \frac{\lambda_{\min}^2}{\lambda_i(1-\lambda_{\min})}~(*)$     \\
 \hline
 Tor($A'$) & $\prod^M_{i=1}G_i^T(r_i,n)~(*)$~Eq.~(\ref{eq:uptor})&$\prod_{i}^M L^T_i(r_i,n)~(*)$~Eq.~(\ref{eq:lowtor})  \\
\end{tabularx}
\end{ruledtabular}
}
\caption{Upper and lower bounds on various matrix functions.}
\label{table:upperlower}
\end{table*}
\subsection{Permanent of Hermitian positive semidefinite matrices}
From Eq.~(\ref{eq:pergbs}), the permanent of an HPSD matrix is connected to a GBS circuit with a thermal state input. The probability of all single-photon measurements is written as
\begin{align}\label{eq:perlower}
     p_\text{th}&=\int d^{2M}\bm{\alpha} \prod^M_{i=1}\frac{2}{\pi(2n_i+1-s)} e^{-\frac{2}{(2n_i+1-s)}|\alpha_i|^2} \prod_{j=1}^M \frac{8|\beta_j|^2+2(s^2-1)}{(s+1)^3} e^{-\frac{2}{s+1}|\beta_j|^2} \\
     &\leq \frac{2}{\prod_{i=1}^M\pi(2n_i+1-s)}\int d^{2M}\bm{\alpha} \prod^M_{i=1}e^{-\frac{2}{2n_{\max}+1-s}|\alpha_i|^2} \prod_{j=1}^M \frac{8|\beta_j|^2+2(s^2-1)}{(s+1)^3} e^{-\frac{2}{s+1}|\beta_j|^2} \\
     &=\frac{2}{\prod_{i=1}^M\pi(2n_i+1-s)}\int d^{2M}\bm{\beta}\prod_{j=1}^M\frac{8|\beta_j|^2+2(s^2-1)}{(s+1)^3} e^{-\left(\frac{2}{s+1}+\frac{2}{2n_{\max}+1-s}\right)|\beta_j|^2} \\
     &=\prod_{i=1}^M(1-\lambda_i)\frac{\lambda_{\max}\{\lambda_{\max}(2+\lambda_{\min})-2-3\lambda_{\min}\}}{\lambda_i(2+\lambda_{\min})-2-3\lambda_{\min}},
\end{align}
where the inequality is valid when $s \geq 1$. Note that we take $s=2n_{\min}+1$ and $n_i=\frac{\lambda_i}{1-\lambda_i}$ in the last equality.

Thus the permanent of an HPSD matrix $B$ are bounded from above as
\begin{equation}\label{eq:upper}
    \Per(B)=p_{\text{th}}{\cal Z'} \leq \prod_{i=1}^M\frac{\lambda_{\max}\{\lambda_{\max}(2+\lambda_{\min})-2-3\lambda_{\min}\}}{\lambda_i(2+\lambda_{\min})-2-3\lambda_{\min}}\coloneqq \prod_{i=1}^M G_i(\lambda_i).
\end{equation}
For the lower bound, our method gives the same result in Ref.~\cite{chakhmakhchyan2017quantum}.

\subsection{Hafnian}
Let us set $a_{i,\pm}(r_i,n)=(2n+1)e^{\pm 2r_i}$, $a_{\min}=(2n+1)e^{-2r_{\max}}$, and $a_{\max}=(2n+1)e^{2r_{\max}}$. When the input state is a squeezed thermal state $\{r_i,n\}_{i=1}^M$ with a fixed $n$, the probability of all single-photon detection is
\begin{align}
   p_{\text{st}}&=\int d^{2M}\bm{\alpha} \prod^{M}_{i=1}\frac{2}{\pi \sqrt{(a_+(r_i,n)-s)(a_-(r_i,n)-s)}}e^{-\frac{2\alpha_{ix}^2}{a_+(r_i,n)-s}-\frac{2\alpha_{iy}^2}{a_-(r_i,n)-s}} \prod_{j=1}^{M} \frac{8|\beta_j|^2+2(s^2-1)}{(s+1)^3} e^{-\frac{2|\beta_j|^2}{s+1}} \\
   &\geq \frac{2}{\prod^M_{i=1}\pi \sqrt{(a_+(r_i,n)-s)(a_-(r_i,n)-s)}}\int d^{2M}\bm{\alpha} \prod^{M}_{i=1}e^{-\frac{2|\alpha_{i}|^2}{a_{\min}-s}} \prod_{j=1}^{M} \frac{8|\beta_j|^2+2(s^2-1)}{(s+1)^3} e^{-\frac{2|\beta_j|^2}{s+1}} \\
   &=\frac{2}{\prod^M_{i=1}\pi \sqrt{(a_+(r_i,n)-s)(a_-(r_i,n)-s)}}\int d^{2M}\bm{\beta} \prod_{j=1}^{M} \frac{8|\beta_j|^2+2(s^2-1)}{(s+1)^3} e^{-\left(\frac{2}{s+1}+\frac{2}{a_{\min}-s}\right)|\beta_{j}|^2}\\
   &=\prod_{i=1}^M  \frac{2 \left(-2 n+e^{2 r_{\max}}-1\right)^2}{\left(2 n+e^{2 r_{\max}}+1\right)^2 \sqrt{-2 (2 n+1) \cosh (2 r_i)+4 n (n+1)+2}},
\end{align}
where the inequality holds for $s \geq 1$, and we put $s=1$ for the last equality. Here, assume $a_{\min}\geq 1$. Then for a matrix $A=\begin{pmatrix} R & B \\ B^T & R^*\end{pmatrix}$ satisfying Eqs.~(\ref{eq:acondition1}) and (\ref{eq:acondition2}),
a lower bound of the Hafnian is written as
\begin{align}\label{eq:lowhaf}
    &\Haf(A)=p_{\text{st}}{\cal Z''}\geq \prod_{i=1}^M \left(\frac{-2 n+e^{2 r_{\max}}-1}{2 n+e^{2 r_{\max}}+1}\right)^2\frac{\sqrt{\frac{1}{2}+n(n+1)+(n+\frac{1}{2})\cosh{2r_i}}}{\sqrt{\frac{1}{2}+n(n+1)-(n+\frac{1}{2})\cosh{2r_i}}}\coloneqq \prod_{i=1}^M L_i^H(r_i,n),\\
    &{\cal Z}''=\sqrt{|V_Q|}=\prod_{i=1}^M \sqrt{\frac{1}{2}+n(n+1)+(n+\frac{1}{2})\cosh{2r_i}}.
\end{align}
Similarly, an upper bound can be obtained as
\begin{align}
   p_{\text{st}}&=\int d^{2M}\bm{\alpha} \prod^{M}_{i=1}\frac{2}{\pi \sqrt{(a_+(r_i,n)-s)(a_-(r_i,n)-s)}}e^{-\frac{2\alpha_{ix}^2}{a_+(r_i,n)-s}-\frac{2\alpha_{iy}^2}{a_-(r_i,n)-s}} \prod_{j=1}^{M} \frac{8|\beta_j|^2+2(s^2-1)}{(s+1)^3} e^{-\frac{2|\beta_j|^2}{s+1}} \\
   &\leq \frac{2}{\prod^M_{i=1}\pi \sqrt{(a_+(r_i,n)-s)(a_-(r_i,n)-s)}}\int d^{2M}\bm{\alpha} \prod^{M}_{i=1}e^{-\frac{2|\alpha_{i}|^2}{a_{\max}-s}} \prod_{j=1}^{M} \frac{8|\beta_j|^2+2(s^2-1)}{(s+1)^3} e^{-\frac{2|\beta_j|^2}{s+1}} \\
   &=\frac{2}{\prod^M_{i=1}\pi \sqrt{(a_+(r_i,n)-s)(a_-(r_i,n)-s)}}\int d^{2M}\bm{\beta} \prod_{j=1}^{M} \frac{8|\beta_j|^2+2(s^2-1)}{(s+1)^3} e^{-\left(\frac{2}{s+1}+\frac{2}{a_{\max}-s}\right)|\beta_{j}|^2}\\
   &=\prod_{i=1}^M \left(\frac{e^{r_{\max}}n+\sinh{r_{\max}}}{(1+n)\cosh{r_{\max}}+n\sinh{r_{\max}}}\right)^2\frac{1}{\sqrt{\frac{1}{2}+n(n+1)-(n+\frac{1}{2})\cosh{2r_i}}},
\end{align}
where we put $s=1$ for the last inequality. Consequently, an upper bound is obtained as
\begin{equation}\label{eq:uphaf}
\Haf(A)=p_{\text{st}}{\cal Z''}\leq \prod_{i=1}^M \left(\frac{e^{r_{\max}}n+\sinh{r_{\max}}}{(1+n)\cosh{r_{\max}}+n\sinh{r_{\max}}}\right)^2\frac{\sqrt{\frac{1}{2}+n(n+1)+(n+\frac{1}{2})\cosh{2r_i}}}{\sqrt{\frac{1}{2}+n(n+1)-(n+\frac{1}{2})\cosh{2r_i}}}\coloneqq \prod_{i=1}^M G_i^H(r_i,n).
\end{equation}

\subsection{Torontonian}
First, we consider thermal state input $\{n_i\}_{i=1}^M$ and the probability of all ``click" outcomes is written as
\begin{align}
    p_{\text{on}|\text{th}}&=\int d^{2M}\bm{\alpha} \prod^{M}_{i=1}\frac{2}{\pi(2n_i+1-s)} e^{-\frac{2}{2n_i+1-s}|\alpha_i|^2}\prod_{j=1}^{M} \left( 1-\frac{2}{s+1}e^{-\frac{2}{s+1}|\beta_j|^2}\right)\\
    &\geq \frac{2}{\prod_{i=1}^M\pi(2n_i+1-s)}\int d^{2M}\bm{\alpha} \prod^M_{i=1}e^{-\frac{2}{2n_{\min}+1-s}|\alpha_i|^2} \prod_{j=1}^{M} \left( 1-\frac{2}{s+1}e^{-\frac{2}{s+1}|\beta_j|^2}\right) \\
    &=\frac{2}{\prod_{i=1}^M\pi(2n_i+1-s)}\int d^{2M}\bm{\beta}  \prod_{j=1}^{M} \left( 1-\frac{2}{s+1}e^{-\frac{2}{s+1}|\beta_j|^2}\right)e^{-\frac{2}{2n_{\min}+1-s}|\beta_j|^2} \\
    &=\prod_{i=1}^M \frac{(1-\lambda_i)\lambda_{\min}^2}{\lambda_i(1-\lambda_{\min})},
\end{align}
where the inequality is valid when $s\geq 1$, and we take $s=1$ for the last equality. Thus for an $M \times M$ HPSD matrix $B$, the Torontonian of $\begin{pmatrix}B^T &0\\0 &B \end{pmatrix}$ is
\begin{equation}
    \text{Tor}\begin{pmatrix}B^T &0\\0 &B \end{pmatrix} = p_{\text{on}|\text{th}}{\cal Z'}\geq\prod_{i=1}^M \frac{\lambda_{\min}^2}{\lambda_i(1-\lambda_{\min})}.
\end{equation}

Similarly, we also obtain an upper bound as
\begin{align}
    p_{\text{on}|\text{th}}&=\int d^{2M}\bm{\alpha} \prod^{M}_{i=1}\frac{2}{\pi(2n_i+1-s)} e^{-\frac{2}{2n_i+1-s}|\alpha_i|^2}\prod_{j=1}^{M} \left( 1-\frac{2}{s+1}e^{-\frac{2}{s+1}|\beta_j|^2}\right)\\
    &\leq \frac{2}{\prod_{i=1}^M\pi(2n_i+1-s)}\int d^{2M}\bm{\alpha} \prod^M_{i=1}e^{-\frac{2}{2n_{\max}+1-s}|\alpha_i|^2} \prod_{j=1}^{M} \left( 1-\frac{2}{s+1}e^{-\frac{2}{s+1}|\beta_j|^2}\right) \\
    &=\frac{2}{\prod_{i=1}^M\pi(2n_i+1-s)}\int d^{2M}\bm{\beta}  \prod_{j=1}^{M} \left( 1-\frac{2}{s+1}e^{-\frac{2}{s+1}|\beta_j|^2}\right)e^{-\frac{2}{2n_{\max}+1-s}|\beta_j|^2} \\
    &=\prod_{i=1}^M \frac{(1-\lambda_i)\lambda_{\max}^2}{\lambda_i(1-\lambda_{\max})},
\end{align}
where we take $s=1$ for the last equality. Consequently,
\begin{equation}
    \text{Tor}\begin{pmatrix}B^T &0\\0 &B \end{pmatrix} = p_{\text{on}|\text{th}}{\cal Z'}\leq \prod_{i=1}^M \frac{\lambda_{\max}^2}{\lambda_i(1-\lambda_{\max})}.
\end{equation}
Next, when the input state is a squeezed thermal state $\{r_i,n\}_{i=1}^M$. Then the probability of all ``click'' detection is
\begin{align}
   p_{\text{on}|\text{st}}&=\int d^{2M}\bm{\alpha} \prod^{M}_{i=1}\frac{2}{\pi \sqrt{(a_+(r_i,n)-s)(a_-(r_i,n)-s)}}e^{-\frac{2\alpha_{ix}^2}{a_+(r_i,n)-s}-\frac{2\alpha_{iy}^2}{a_-(r_i,n)-s}} \prod_{j=1}^{M} \left( 1-\frac{2}{s+1}e^{-\frac{2}{s+1}|\beta_j|^2}\right) \\
   &\geq \frac{2}{\prod^M_{i=1}\pi \sqrt{(a_+(r_i,n)-s)(a_-(r_i,n)-s)}}\int d^{2M}\bm{\alpha} \prod^{M}_{i=1}e^{-\frac{2|\alpha_{i}|^2}{a_{\min}-s}} \prod_{j=1}^{M} \left( 1-\frac{2}{s+1}e^{-\frac{2}{s+1}|\beta_j|^2}\right) \\
   &=\frac{2}{\prod^M_{i=1}\pi \sqrt{(a_+(r_i,n)-s)(a_-(r_i,n)-s)}}\int d^{2M}\bm{\beta} \prod_{j=1}^{M} \left( 1-\frac{2}{s+1}e^{-\frac{2}{s+1}|\beta_j|^2}\right)e^{-\frac{2|\beta_{j}|^2}{a_{\min}-s}}\\
   &=\prod_{i=1}^M \frac{e^{-2r_{\max}}(1-e^{2r_{\max}}+2n)^2}{2(1+e^{2r_{\max}}+2n)}\frac{1}{\sqrt{\frac{1}{2}+n(n+1)-(n+\frac{1}{2})\cosh{2r_i}}},
\end{align}
where $a_{i,\pm}(r_i,n)=(2n+1)e^{\pm 2r_i}$, $a_{\min}=(2n+1)e^{-2r_{\max}}$, $a_{\max}=(2n+1)e^{2r_{\max}}$, and we put $s=1$ for the last inequality. Assume $a_{\min}\geq 1$. Then for a matrix $A'=\begin{pmatrix} B^T & R^* \\R & B\end{pmatrix}$ satisfying Eqs.~(\ref{eq:acondition1}) and (\ref{eq:acondition2}), a lower bound of the Torontonian is given by
\begin{equation}\label{eq:lowtor}
    \text{Tor}(A')= p_{\text{on}|\text{st}}{\cal Z}'' \geq \prod_{i=1}^M \frac{e^{-2r_{\max}}(1-e^{2r_{\max}}+2n)^2}{2(1+e^{2r_{\max}}+2n)}\frac{\sqrt{\frac{1}{2}+n(n+1)+(n+\frac{1}{2})\cosh{2r_i}}}{\sqrt{\frac{1}{2}+n(n+1)-(n+\frac{1}{2})\cosh{2r_i}}}\coloneqq \prod^M_{i=1}L_i^T(r_i,n).
\end{equation}

An upper bound can be obtained by a similar method, such as
\begin{align}
   p_{\text{on}|\text{st}}&=\int d^{2M}\bm{\alpha} \prod^{M}_{i=1}\frac{2}{\pi \sqrt{(a_+(r_i,n)-s)(a_-(r_i,n)-s)}}e^{-\frac{2\alpha_{ix}^2}{a_+(r_i,n)-s}-\frac{2\alpha_{iy}^2}{a_-(r_i,n)-s}} \prod_{j=1}^{M} \left( 1-\frac{2}{s+1}e^{-\frac{2}{s+1}|\beta_j|^2}\right) \\
   &\leq \frac{2}{\prod^M_{i=1}\pi \sqrt{(a_+(r_i,n)-s)(a_-(r_i,n)-s)}}\int d^{2M}\bm{\alpha} \prod^{M}_{i=1}e^{-\frac{2|\alpha_{i}|^2}{a_{\max}-s}} \prod_{j=1}^{M} \left( 1-\frac{2}{s+1}e^{-\frac{2}{s+1}|\beta_j|^2}\right) \\
   &=\frac{2}{\prod^M_{i=1}\pi \sqrt{(a_+(r_i,n)-s)(a_-(r_i,n)-s)}}\int d^{2M}\bm{\beta} \prod_{j=1}^{M} \left( 1-\frac{2}{s+1}e^{-\frac{2}{s+1}|\beta_j|^2}\right)e^{-\frac{2|\beta_{j}|^2}{a_{\max}-s}}\\
   &=\prod_{i=1}^M \frac{e^{r_{\max}}(e^{r_{\max}}n+\sinh{r_{\max}})^2}{(1+n)\cosh{r_{\max}}+n\sinh{r_{\max}}}\frac{1}{\sqrt{\frac{1}{2}+n(n+1)-(n+\frac{1}{2})\cosh{2r_i}}},
\end{align}
where we put $s=1$ for the last inequality. Consequently, an upper bound of the $\text{Tor}(A')$ is obtained as
\begin{equation}\label{eq:uptor}
    \text{Tor}(A')= p_{\text{on}|\text{st}}{\cal Z}'' \leq \prod_{i=1}^M \frac{e^{r_{\max}}(e^{r_{\max}}n+\sinh{r_{\max}})^2}{(1+n)\cosh{r_{\max}}+n\sinh{r_{\max}}}\frac{\sqrt{\frac{1}{2}+n(n+1)+(n+\frac{1}{2})\cosh{2r_i}}}{\sqrt{\frac{1}{2}+n(n+1)-(n+\frac{1}{2})\cosh{2r_i}}}\coloneqq \prod^M_{i=1}G_i^T(r_i,n).
\end{equation}

\section{Simulability of Gaussian boson sampling (Proof of Theorem~\ref{th:GBS})}\label{app:GBSsimul}
Our approximation algorithm for outcome probabilities of a linear optical circuit have applications not only for the matrix functions, but also for Gaussian boson sampling, which is crucial for the demonstration of quantum supremacy~\cite{hamilton2017gaussian}.
From the results in Refs.~\cite{pashayan2020estimation}, we have three level of a hierarchy of notions of classical simulation as following:
\begin{enumerate}
    \item {Poly-box}
    : Inverse-polynomial additive-error approximation of any outcome probabilities including any marginals.
    \item {$\epsilon$-simulation}
    : Approximate sampling simulation of probability distributions with $\epsilon$-error in total variation distance.
    \item {Multiplicative precision estimation}
    : multiplicative-error approximation of any outcome probabilities including any marginals.
\end{enumerate}

A poly-box can be promoted to $\epsilon$-simulation when the outcomes are poly-sparse, and multiplicative precision estimator implies $\epsilon$-simulation~\cite{pashayan2020estimation}.  
In our work, we can investigate this hierarchy in the GBS via a degree of classicality of the input state, $s_{\max}$. To do that, we consider a lossy GBS, in which the input state is product of lossy squeezed state having the covariance matrix on $i$th mode with $V_i=\frac{1}{2}\begin{pmatrix}\eta e^{2r_i}+1-\eta & 0 \\ 0 & \eta e^{-2r_i}+1-\eta  \end{pmatrix}\coloneqq  \frac{1}{2}\begin{pmatrix}a_{i+}(\eta,r_i) & 0 \\ 0 & a_{i-}(\eta,r_i)\end{pmatrix}$. From Eq.~(\ref{eq:spqd}), $-1<s\leq s_{\max}\leq 1$ for a lossy squeezed state, and $s_{\max}=a_-(\eta,r_{\max})$. Note that the $s_{\max}$ goes to $0$ as the maximum squeezing parameter $r_{\max} \rightarrow \infty$, which is consistent with the fact that a general Gaussian state can be well described by Wigner distribution.
Then the outcome probability $p_{\text{GBS}}(\bm{m})$ is given by
\begin{align}
    &p_{\text{GBS}}(\bm{m})=\pi^M \int d^{2M}\bm{\alpha} \prod_{i=1}^M  W^{(s)}_{V_{i}}(\alpha_i)\prod_{j=1}^M W^{(-s)}_{\Pi_{m_j}}(\beta_j) \\
     &=\pi^M\int d^{2M}\bm{\alpha}\prod_{i=1}^M \frac{1}{\pi \sqrt{\det{(V_{i}-s/2)}}}e^{-\bm{\alpha}_i {(V_{i}-s/2)}^{-1}\bm{\alpha}_i^T}\prod_{j=1}^M \frac{2}{\pi(s+1)}\left(\frac{s-1}{s+1} \right)^{m_j} \text{L}_{m_j} \left(\frac{4|\beta_j|^2}{1-s^2} \right) e^{-\frac{2|\beta_j|^2}{s+1}}\\
     &=\pi^M\int d^{2M}\bm{\alpha} \prod_{i=1}^M \frac{2}{\pi \sqrt{(a_{i+}(\eta,r_i)-s)(a_{i-}(\eta,r_i)-s)}}e^{-\frac{2\alpha_{ix}^2}{a_{i+}(\eta,r_i)-s}-\frac{2\alpha_{iy}^2}{a_{i-}(\eta,r_i)-s}}\prod_{j=1}^M \frac{2}{\pi(s+1)}\left(\frac{s-1}{s+1} \right)^{m_j} \text{L}_{m_j} \left(\frac{4|\beta_j|^2}{1-s^2} \right) e^{-\frac{2|\beta_j|^2}{s+1}}\\
     &=\int d^{2M}\bm{\alpha}\prod^{M}_{i=1}\frac{2}{\pi} \sqrt{\frac{a_{+}(\eta,r_{\max})-s-\gamma(a_{i+}(\eta,r_i)-s)}{(a_{i+}(\eta,r_i)-s)(a_{+}(\eta,r_{\max})-s)}}\sqrt{\frac{a_{+}(\eta,r_{\max})-s-\gamma(a_{i-}(\eta,r_i)-s)}{(a_{i-}(\eta,r_i)-s)(a_{+}(\eta,r_{\max})-s)}}\\
     &\times e^{-\left(\frac{2}{a_{i+}(\eta,r_i)-s}-\gamma \frac{2}{a_{+}(\eta,r_{\max})-s} \right)\alpha_{ix}^2-\left(\frac{2}{a_{i-}(\eta,r_i)-s}-\gamma \frac{2}{a_{+}(\eta,r_{\max})-s} \right)\alpha_{iy}^2}\\
    &\times \prod_{j=1}^{M} \sqrt{\frac{a_{+}(\eta,r_{\max})-s}{a_{+}(\eta,r_{\max})-s-\gamma(a_{i+}(\eta,r_i)-s)}}\sqrt{\frac{a_{+}(\eta,r_{\max})-s}{a_{+}(\eta,r_{\max})-s-\gamma(a_{i-}(\eta,r_i)-s)}}\\
    &\times \frac{2}{s+1}\left(\frac{s-1}{s+1} \right)^{m_j} \text{L}_{m_j} \left(\frac{4|\beta_j|^2}{1-s^2} \right) e^{-\left(\frac{2}{s+1}+\gamma \frac{2}{a_{+}(\eta,r_{\max})-s}\right)|\beta_j|^2} \\
    &\coloneqq \int d^{2M}\bm{\alpha}\prod^M_{i=1}P_{\text{GBS},i}(\alpha_i,\eta,r_i,\gamma,s)\prod^M_{j=1}f_{\text{GBS},j}(\beta_j,\eta,r_j,m_j,\gamma,s),
\end{align}
where $a_{\pm}(\eta,r_{\max})=\eta e^{\pm 2r_{\max}}+1-\eta$, and $\gamma \in (0,1]$ is a parameter modulating the Gaussian factor such as $\gamma \rightarrow 1$ ($\gamma=0$) means maximum (no) shifting. The maximum shifting is limited by the maximum squeezing parameter $r_{\max}$. To check the poly-box condition, let us first consider the single-mode estimate for the single-photon outcome.  Explicitly,
\begin{align}\label{eq:gbs1}
    f_{\text{GBS},j}(\beta_j,\eta,r_j,1,\gamma,s)&=
    \sqrt{\frac{a_{+}(\eta,r_{\max})-s}{a_{+}(\eta,r_{\max})-s-\gamma(a_{j+}(\eta,r_j)-s)}}\sqrt{\frac{a_{+}(\eta,r_{\max})-s}{a_{+}(\eta,r_{\max})-s-\gamma(a_{j-}(\eta,r_j)-s)}} \nonumber \\
    &\times \frac{8|\beta_j|^2+2(s^2-1)}{(s+1)^3} e^{-\left(\frac{2}{s+1}+\gamma \frac{2}{a_{+}(\eta,r_{\max})-s}\right)|\beta_j|^2} .
\end{align}
We can efficiently estimate the probability if $\max_{\beta_j} |f_{\text{GBS},j}| \leq 1$ for all $j$. An upper bound of the absolute value of the estimate is given by
\begin{equation}
    \min_{s,\gamma}\max_{\beta_j}|f_{\text{GBS},j}(\beta_j,\eta,r_j,1,\gamma,s)|\leq \max_{\beta_j}|f_{\text{GBS},j}(\beta_j,\eta,r_{\max},1,0,s_{\max})|,
\end{equation}
for given $\eta,r_{\max}$, and $\gamma=0$, $s=s_{\max}$ for the inequality.
Then from the condition $\max_{\beta_j}|f_j(\beta_j,\eta,r_{\max},1,0,s_{\max})| \leq 1$, $s_{\max}$ satisfies $s_{\max}\geq \sqrt{5}-2\simeq 0.236$. This corresponds to $r_{\max}\leq \frac{1}{2}\log(2+\sqrt{5})\simeq 0.722$ for an ideal GBS ($\eta=1$). However, if we allow the photon loss, any squeezed input state is possible when $\eta \leq 3-\sqrt{5}\simeq 0.764$, which is much higher transimissivity than those used in current experiments~\cite{zhong2021phase,madsen2022quantum}. Next, we need to check whether this condition is valid for any other outcomes. From the behavior of $f_j(\beta,\eta,r,m,0,s)$, we can find out that 
\begin{equation}
    \max_{\beta} |f_j(\beta,\eta,r,m,0,s)| \leq \max_{\beta} |f_j(\beta,\eta,r,1,0,s)|,
\end{equation}
for $m \geq 2$ and $s\geq 0$. Finally, we consider $n=0$ for zero-photon detection and $f_j=1$ for the marginalized probability owing to the normalization of measurement operators. In both cases, the integrals for $\beta_j$'s can be easily computed because $f_j(\beta_j)$ and $\beta_j$ components in $P(\bm{\alpha})$ are just Gaussian distributions. Therefore, we can always perform the integrals including $\beta_j$'s corresponding to zero-photon or marginalized one, and estimate remaining terms.
Furthermore, we examine the case of threshold detectors instead of number resolving measurements~\cite{quesada2018gaussian}. The corresponding $f_{\text{on},j}$ for a `click' event is written as
\begin{align}\label{eq:GBSth}
    f_{\text{on},j}(\beta_j,\eta,r_j,\gamma,s)&= \sqrt{\frac{a_{+}(\eta,r_{\max})-s}{a_{+}(\eta,r_{\max})-s-\gamma(a_{j+}(\eta,r_j)-s)}}\sqrt{\frac{a_{+}(\eta,r_{\max})-s}{a_{+}(\eta,r_{\max})-s-\gamma(a_{j-}(\eta,r_j)-s)}} \nonumber\\
    &\times \left( 1-\frac{2}{s+1}e^{-\frac{2}{s+1}|\beta_j|^2}\right)e^{-\gamma \frac{2}{a_{+}(\eta,r_{\max})-s}|\beta_j|^2}.
\end{align}
Then for $\gamma=0$ and $s=0$, the range of $f_{\text{on},j}(\beta_j,\eta,r_j,0,0)=1-2e^{-2|\beta_j|^2}$ is on $[-1,1]$, thus the poly-box condition is satisfied for all input squeezing and loss parameter. 

Now we investigate whether an efficient estimation of GBS probability within multiplicative-error is possible. To do that, we consider Gaussian states which can have  $s_{\max}>1$, where the covariance matrix of $i$th mode state is given by 
\begin{equation}
    V_i=\frac{1}{2}\begin{pmatrix}\eta e^{2r_i}+(2n_{\text{th}}+1)(1-\eta) & 0 \\ 0 & \eta e^{-2r_i}+(2n_{\text{th}}+1)(1-\eta)  \end{pmatrix}\coloneqq  \frac{1}{2}\begin{pmatrix}a_{i+}(\eta,n_{\text{th}},r_i) & 0 \\ 0 & a_{i-}(\eta,n_{\text{th}},r_i)\end{pmatrix}.
\end{equation}
These are squeezed thermal states, in which pure squeezed states undergo a thermal noise with average photon number $n_{\text{th}}$ instead of the vacuum loss. In this case $-1<s\leq a_-(\eta,n_{\text{th}},r_{\max})$ for given $\eta$, $n_{\text{th}}$, and $r_{\max}$. Then we need to check the log-concavity of $f_{\text{on},j}$ such that
\begin{equation}\label{eq:multith}
    f_{\text{on},j}(\beta_j,\eta,n_{\text{th}},r_i,\gamma,s) \propto \left( 1-\frac{2}{s+1}e^{-\frac{2}{s+1}|\beta_j|^2}\right) e^{-\gamma\frac{2}{a_{+}(\eta,N,r_{\max})-s}|\beta_j|^2}.
\end{equation}

From Lemma~\ref{lem1}, the condition for log-concavity of $f_{\text{on},j}$ when $\gamma \rightarrow 1$ is
\begin{equation}
     \sqrt{4e^{r_{\max}} \eta \sinh{r_{\max}}+4n_{\text{th}}(1-\eta)+5} \leq s \leq a_-(\eta,n_{\text{th}},r_{\max}),
\end{equation}
where $a_-(\eta,n_{\text{th}},r_{\max})=\eta^{-2r_{\max}}+(2n_{\text{th}}+1)(1-\eta)$. Thus for given $r_{\max}$ and $\eta$, the average photon number of thermal noise $N$ satisfies
\begin{equation}
    n_{\text{th}} \geq \frac{e^{-r_{\max}}\eta\sinh{r}+\sqrt{1+\eta\sinh{2r_{\max}}}}{1-\eta}>1.
\end{equation}
For instance, if $\eta=0.5$ and $r_{\max}=1$, then $n_{\text{th}} \geq n_{\text{th}}^* \simeq 3.79$ for the multiplicative-error estimation of the probability, and the minimum value of $s_{\max}$ is $3$ when $\eta \rightarrow 0$.

\end{widetext}

\bibliography{reference}
\end{document}